\documentclass[a4paper,12pt,envcountsect,envcountresetsect,envcountsame]{jpsvmult}
\usepackage{times}
\usepackage{amssymb}
\usepackage{epsfig}
\usepackage{enumerate}
\usepackage{mathrsfs}
\let\rho=\varrho

\let\epsilon=\varepsilon
\def\captionstyle{\small\rmfamily}

\def\sect#1#2{\section{#1}\label{#2}}
\def\subsect#1#2{\subsection{#1}\label{#2}}
\def\sref#1{Sect.~\ref{#1}}
\def\lref#1{Lemma~\ref{#1}}
\def\fref#1{Fig.~\ref{#1}}
\def\eref#1{(\ref{#1})}
\def\T{{\cal T}}
\def\TT{{\mathbb T}}
\def\ZZ{{\mathbb Z}}
\def\natural{{\mathbb N}}
\def\integer{{\mathbb Z}}
\def\OO{{\cal O}}
\def\GG{\cal G}
\def\ie{{\it i.e.}}
\def\eg{{\it e.g.}}
\def\const{{\rm const.}}
\usepackage{MHequ}
  {\setlength\paperheight {297mm}%
    \setlength\paperwidth  {210mm}}

\begin{document}
\title*{A Topological Glass}
\bigskip
\date{}

\author{Jean-Pierre Eckmann}
\institute{D\'epartement de Physique Th\'eorique et\\
Section de Math\'ematiques\\
Universit\'e de Gen\`eve\\
1211 Gen\`eve 4, Switzerland}
\maketitle
\begin{abstract}
  We propose and study a model with glassy behavior. The state space
  of the model is given by all triangulations of a sphere with $n$
  nodes, half of which are red and half are blue. Red nodes want to
  have 5 neighbors while blue ones want 7. Energies of nodes with
  different numbers of neighbors are supposed to be positive. 
The dynamics is that of
  flipping the diagonal of two adjacent triangles, with a temperature
  dependent probability. We show that this system has an approach to a
  steady state which is exponentially slow, and show that the
  stationary state is unordered. We also study the local energy
  landscape and show that it has the hierarchical structure known from
  spin glasses. Finally, we show that the evolution can be described as that of
  a rarefied gas with spontaneous generation of particles and annihilating collisions.
\end{abstract}
\thispagestyle{empty}
\sect{Introduction}{intro}
Our model of a glass can is inspired by 
an abstraction of a model and its
representation found in \cite{Proglass2007}. Before we introduce our version,
and to make contact with that paper, we describe briefly the model of \cite{Proglass2007}.
One starts with a binary mixture of $n$ disks in the plane, with the small
disks having radius $\sigma_{\rm s}=1$ and the large, $\sigma_{\rm \ell}=1.4$. The three pairwise additive
interactions are given by  purely repulsive soft-core potentials 
of the form
$$
\epsilon \left(\frac{\sigma_a + \sigma_b}{2r}\right)^{12}~,
$$
with $a,b\in\{{\rm s},{\rm \ell}\}$, $\epsilon >0$, and $r$ the distance between the
  centers of the disks. One assumes the interaction vanishes for $r>
  2.25(\sigma_a+\sigma_b)$. 
Taking periodic boundary conditions, and a relatively tight volume,
the authors of \cite{Proglass2007} found that this system shows the
characteristics of a glass when the temperature is sufficiently low.

The part of the analysis which is of interest for the present study
has to do with a geometric representation of configurations of
this system. One draws in the plane a point for the
position of the centers of the disks, and proceeds then to use the
Voronoi tessellation. This means that lines are drawn between nearest
neighbors in the sample, and their normal bisectors are then used to
draw polygons around each particle.
\begin{figure}[!ht]
  \centerline{\psfig{file=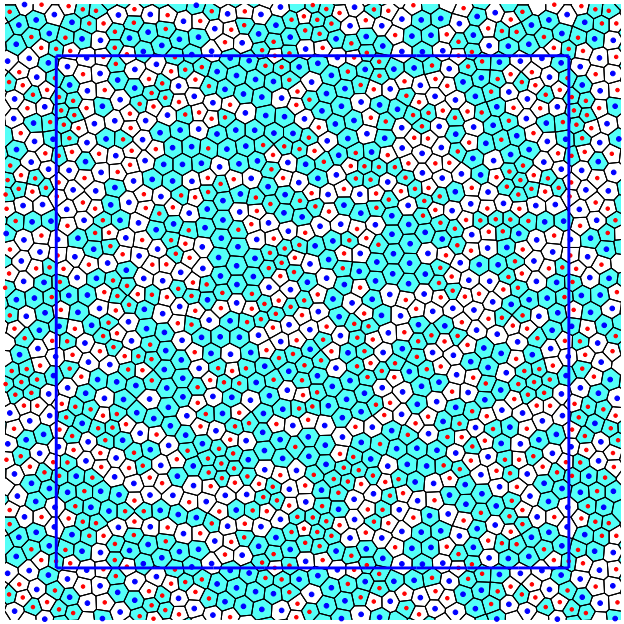,width=8.2cm}}
\caption{A Voronoi tessellation (adapted from \cite{Proglass2007}). The
  interpretation of the figure is as follows: Blue dots correspond to
  the positions of the big particles and red dots to the small ones. The
  polygons are white, for every blue particle in a heptagon and every
  red particle in a pentagon. All hexagons are in blue-green. There
  are no other polygons at this temperature ($T=0.1$, see \cite{Proglass2007}). The blue line shows the
  boundary of the domain, which is extended periodically beyond.
}\label{f:pro}
\end{figure}
These polygons turn out to be mostly pentagons around the small
particles
and heptagons around the large ones, with a few exceptions as a
function of temperature, and by the constraints of Euler's theorem on
polygonal domains on surfaces. See \fref{f:pro}.

In this paper, we consider a purely topological variant of the above
model and show it has the properties of a glass. The model is
basically obtained by considering instead of the Voronoi tessellation
the \emph{dual} graph, which is a triangulation (of the torus). This
triangulation is obtained from \fref{f:pro} by connecting the centers
of the particles normally across the edges of the Voronoi
tessellation.
(The dual graph is a triangulation since each edge of the polygons meets exactly 2 other edges at its
end.)

To simplify things further, we consider instead a triangulation $\T$
of the sphere. Given $n$, the number of ``particles,'' we let $\T$
denote a triangulation of the sphere with $n$ nodes and we let $\TT_{n,0}$
denote the set of 
all such triangulations. By this, one means the set of all combinatorially
distinct rooted simplicial 3-polytopes. In particular, a triangulation
should not have any ``double edges''. We further refine the
definition, by distinguishing 2 types of nodes in the triangulation:
We first number the nodes from 1 to $n$ and then define 
2 types of nodes.
Those with even index are the ``small particles'' and those with odd
index the ``large'' ones.  This means that the triangulation has the same
number of large and small nodes (up to a difference of one). We will
also call the two types of nodes two \emph{colors}.
\begin{definition}
  We shall call the odd nodes \emph{blue} and the even ones \emph{red}
  and will refer to the triangulations as \emph{colored} triangulations.
\end{definition}

Once the colors are assigned, the numbers are again forgotten.
The set
of all colored triangulations with $n$ nodes will be called $\TT_{n}$.
This is our \emph{phase space} and the dynamics is mapping points in this
phase space to other points.

As in the original model, we introduce an energy $E$ for each
triangulation $\T$:
\begin{equ}\label{e:energy}
  E(\T)=\sum_{i=\rm{odd}} (d_i-7)^2 + 
\sum_{i=\rm{even\vphantom{d}}} (d_i-5)^2 ~,
\end{equ}
where $d_i$ is the number of links meeting at node $i$. 
Except for the
constraint given by Euler, that $\sum_{i=1}^n d_i = 6n-12$, the lower
bound of the energy is obviously $E(\T)\ge0$.
\begin{remark}
  The choice of energy is not as ``universal'' as one could wish. In a
  way, it would be more adequate to be able to develop a theory which
  deals with a family of energies, which all have the property that
  the minima are at again at 5 and 7, but which should somehow be independent
  of the details of how the errors are weighted. For example, one
  expects similar results for an energy of the form
 \begin{equ}
  E(\T)=\sum_{i=\rm{odd}} (d_i-7)^2 + 
2\sum_{i=\rm{even\vphantom{d}}} (d_i-5)^4 ~.
\end{equ}
On the positive side, we will see that the hierarchical structure of
the energy landscape indeed does not depend on the details of the
energies, only one their behavior near the quadratic minima.

\end{remark}

We next define a \emph{dynamics} on the set of all triangulations $\TT_{n}$
which is inspired by the motion of the dual to the Voronoi tessellation
of the model in \cite{Proglass2007}.
\begin{itemize}
  
\item[1)] Choose a link at random.
\item[2)] Consider the two triangles touching that link, say ABC and
  BCD (having the common link BC).
If the link AD exists in the triangulation go back to 1). (This can
happen when ABCD form a tetrahedron.)
\item[3)] In principle, we want to flip the link BC and replace it by AD,
\ie, exchange the triangles ABC and BCD
  with ACD and ABD. This operation is called a flip
  \cite{Negami1999}, or a Gross-Varsted move \cite{Malyshev1999} and
  it will transform the triangulation $\T$ to a new one $\T'$. 
\item[4)] If $E(\T')\le E(\T)$ then perform the flip and continue at 1).
\item[5)] If $E(\T')>E(\T)$ then perform a flip with probability
  $e^{-\beta (E(\T')-E(\T))}$ and continue at 1). This is of course a
  typical Monte-Carlo step (at inverse temperature $\beta $).
\item[6)] Continue at 1).
\end{itemize}

\sect{Irreducibility}{irreducibility}

Having defined precisely the algorithm, we first show that the phase
space $\TT_n$ is irreducible.
\begin{lemma}\label{l:irreducibility}
  The action described above defines an irreducible Markov process on $\TT_{n}$ (when $n>7$). This
  means that any 
  configuration can be reached from any other configuration. 
\end{lemma}
\begin{proof}
For the case of the uncolored triangulations, $\TT_{n,0}$, this is a
well-known result \cite{Negami1999,Wagner1936}, see also \cite{CE2005}. In that case, one shows that every
triangulation can be transformed by a sequence of flips to the
``christmas tree'' of \fref{f:christmas1}. Note that this reduction
takes place \emph{without} shifting around the nodes of the outermost
triangle.
A typical sample move to achieve the reduction of the number of links
at the top is shown in \fref{f:moves}.
\begin{figure}[!htb]
 \centerline{\psfig{file=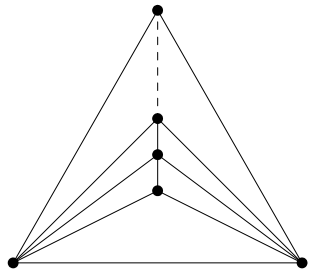,width=4cm}}
\caption{The Christmas tree. Two nodes are at the bottom, the others
  (only 4 shown) are in the stem of the tree.}
\label{f:christmas1}
\end{figure}
\def\y{\kern2em}
\def\x{4cm}
\begin{figure}[!htb]
\centerline{\psfig{file=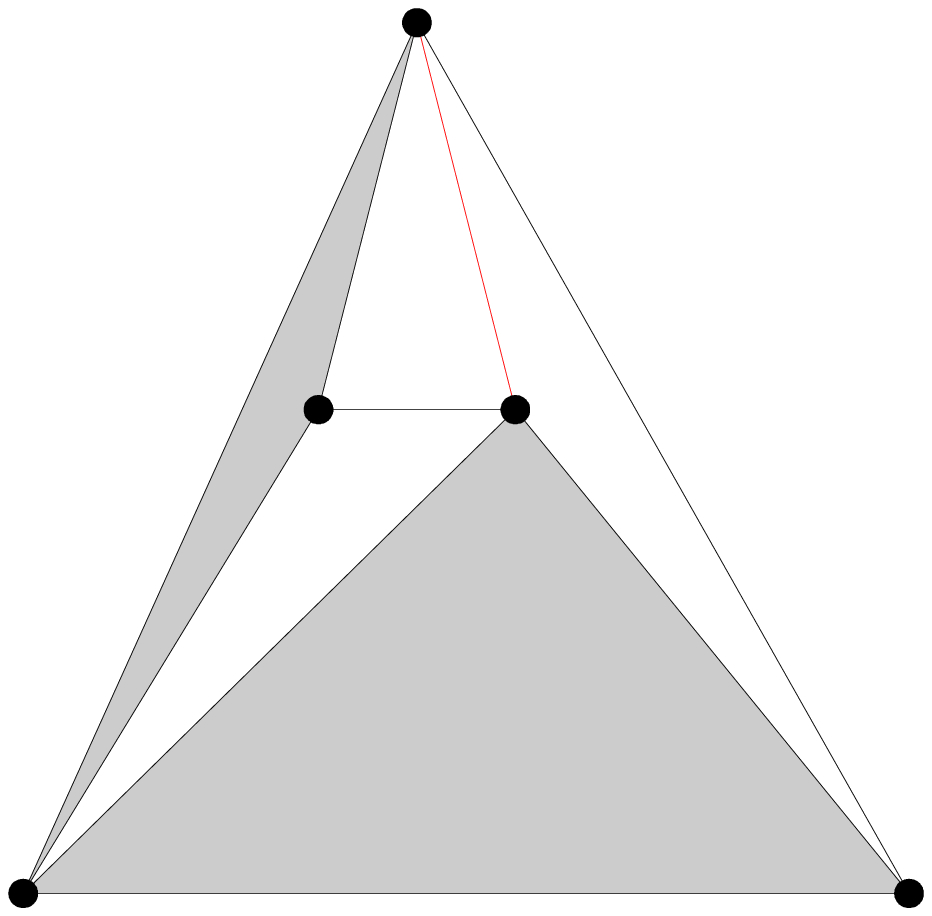,width=\x}\y
\psfig{file=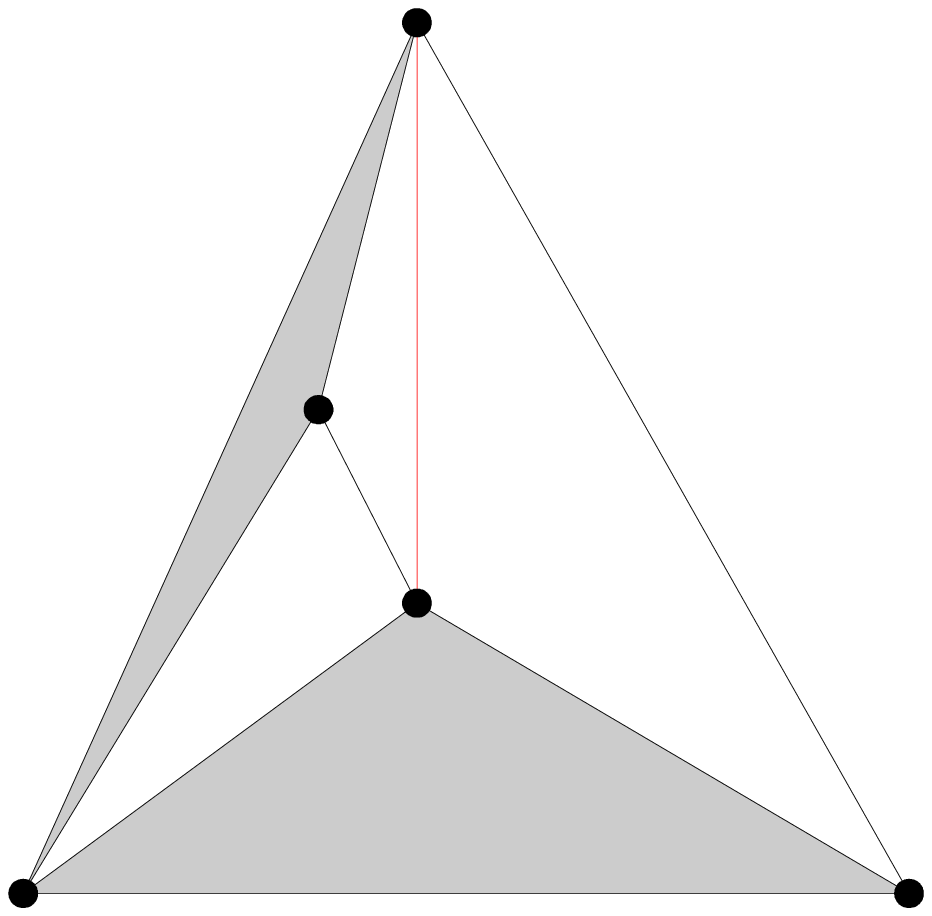,width=\x}}
\medskip
\centerline{\psfig{file=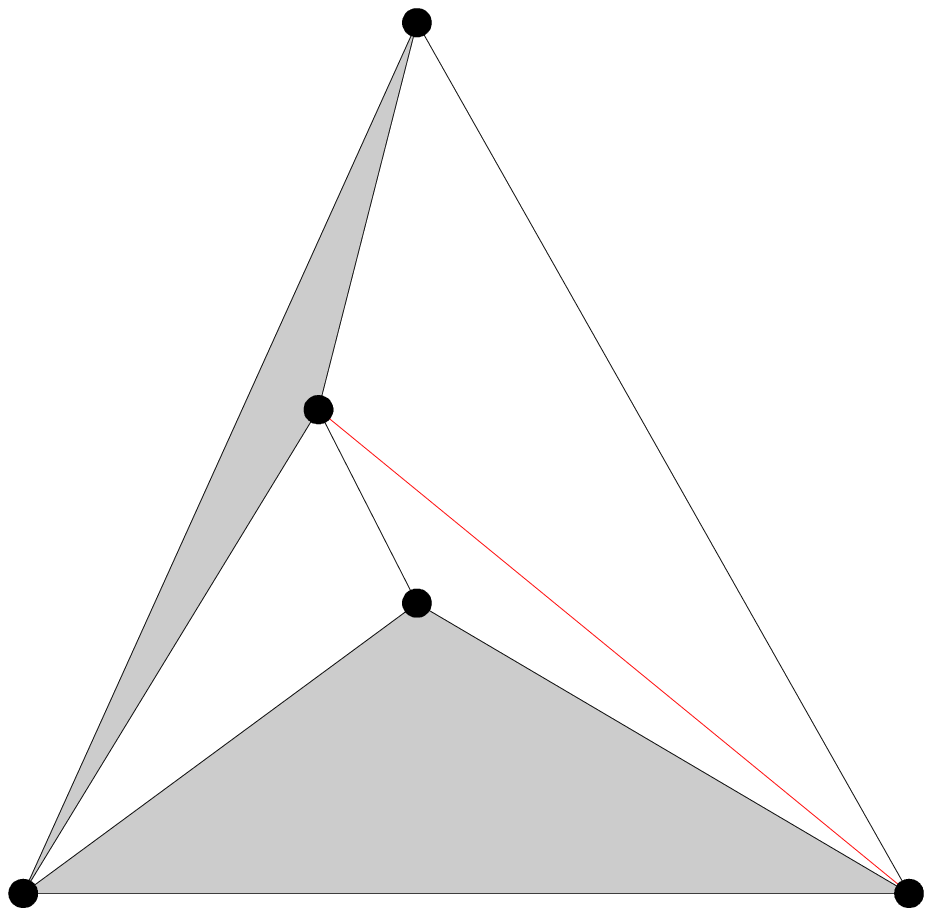,width=\x}\y
\psfig{file=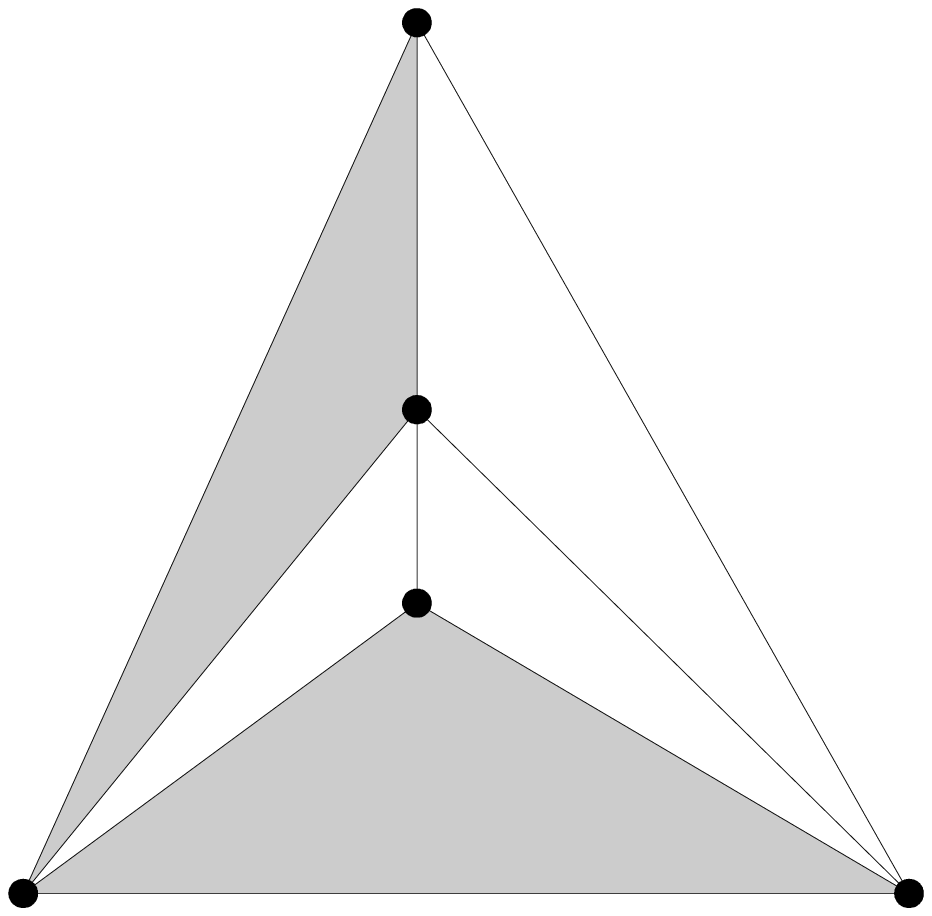,width=\x}
}
\caption{Reducing the degree at the top of the christmas tree by
  1. Note that the gray regions can contain any triangulation. The
  complete christmas tree is obtained by inductively working ``down''
  from the top to the bottom, once the top node has been reduced to
  degree 3. The sequence of events is top left, top
  right, bottom left, and bottom right. The actual flip is done between
  frame 2 and 3. The other transformations just move the nodes into
  place for better visibility.}
  \label{f:moves}
\end{figure}

We next show irreducibility
in $\TT_{n}$, where the positions of
the colors now matter. Since we can move any triangulation to a
christmas tree, we therefore only need to show that one can
exchange the places of the colors in the christmas tree. We first show
in \fref{bluered1} how this is done for any two positions on the stem,
with the exception of the two 
bottom ones. Those are handled by the flips shown in \fref{bluered2}. Finally, one can
``exchange'' colors in the outermost triangle as follows: Observe that
the ``outermost'' triangle is arbitrary since one can declare any
triangle of the triangulation of the sphere to be the outermost
one. In other words, we can reduce the question to the preceding ones
by declaring a new triangle to be the outermost one, bring the
triangulation to the christmas tree form and proceeding as before.\qed
\end{proof}
\begin{remark}
  One can also show, see \cite{CE2005}, that the Markov process
  defined above is aperiodic, \ie, some power of the transition matrix
  has only non-zero entries.
\end{remark}
\def\x{4cm}
\def\y{\kern2em}
\begin{figure}[htb]
  \centerline{
\psfig{file=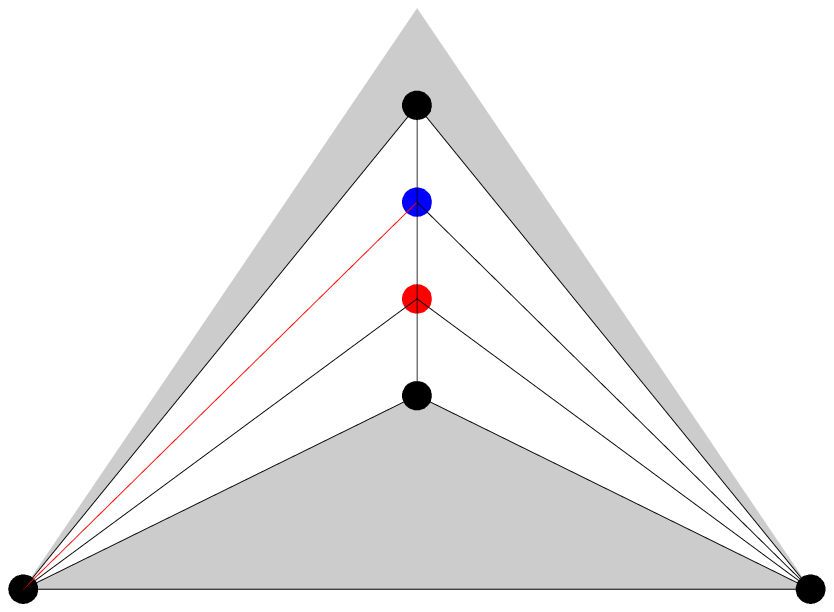,width=\x}\y
\psfig{file=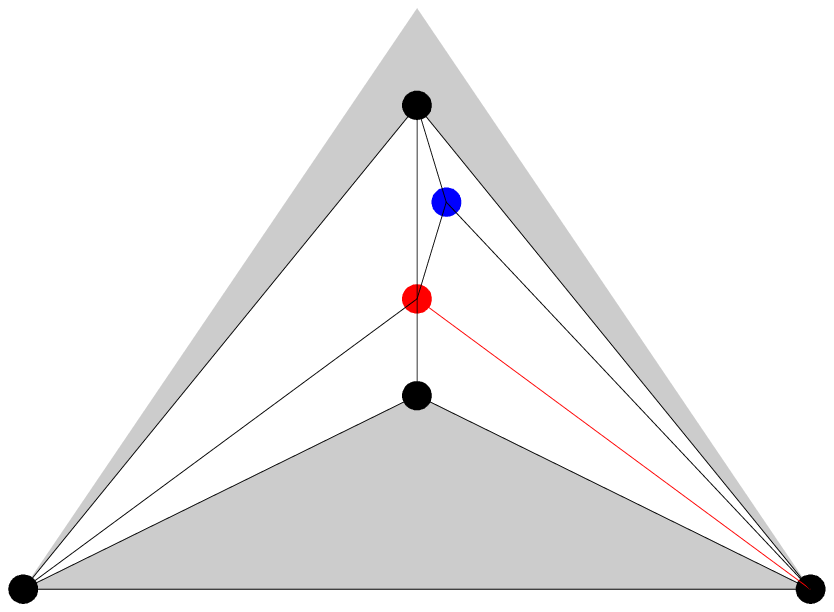,width=\x}\y 
\psfig{file=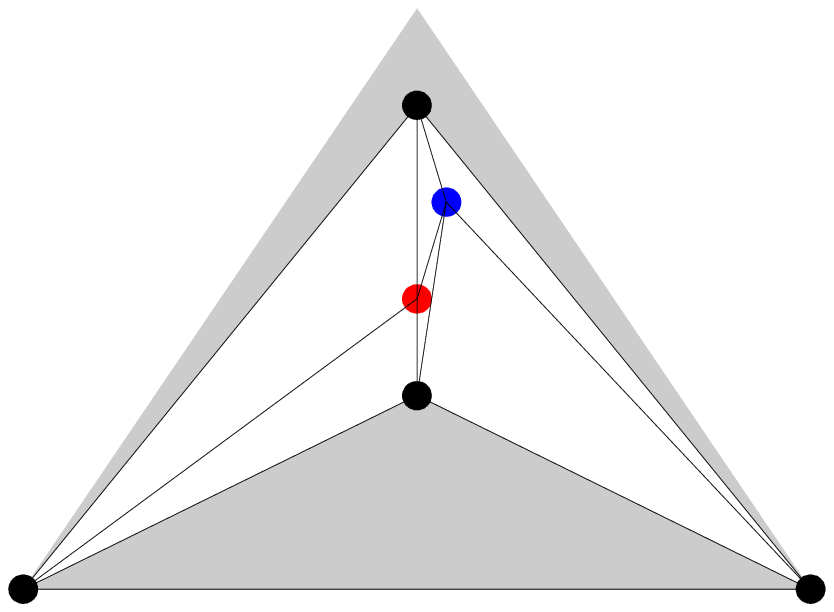,width=\x}}
\centerline{
\psfig{file=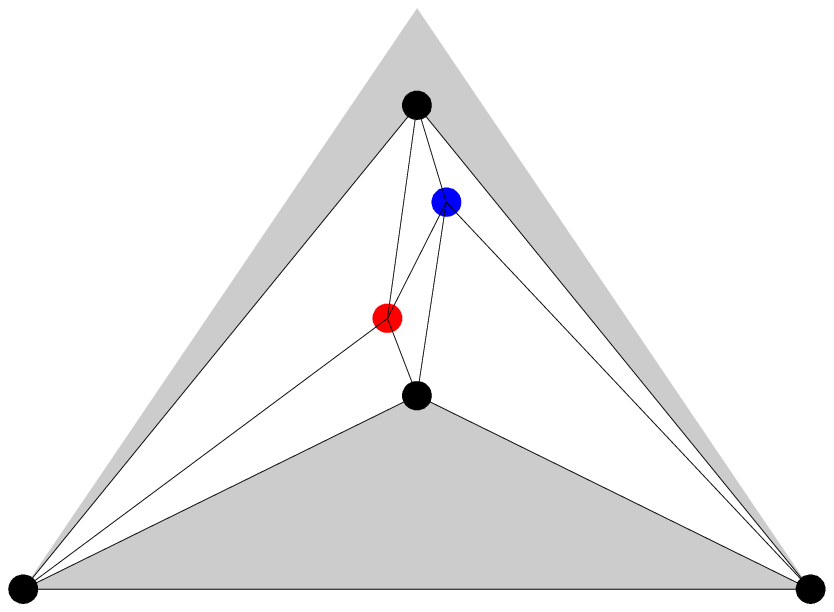,width=\x}\y 
\psfig{file=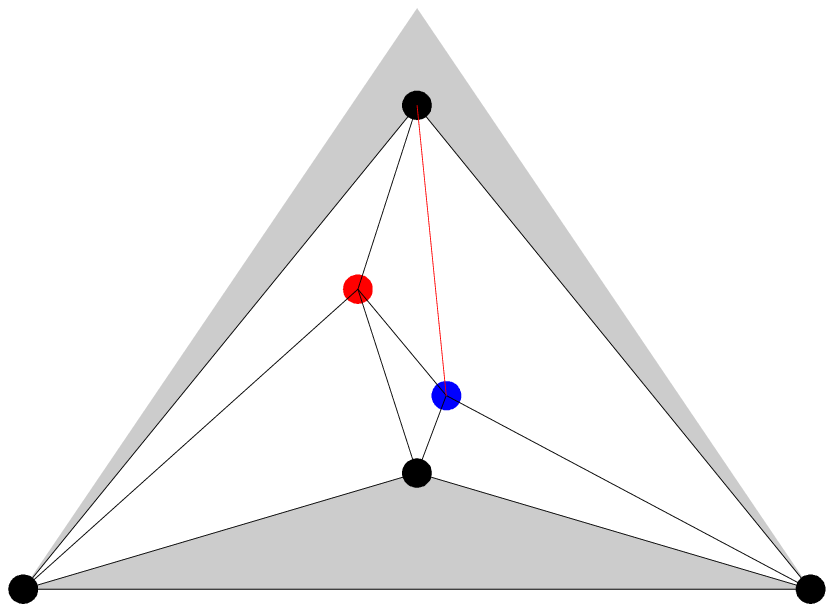,width=\x}\y 
\psfig{file=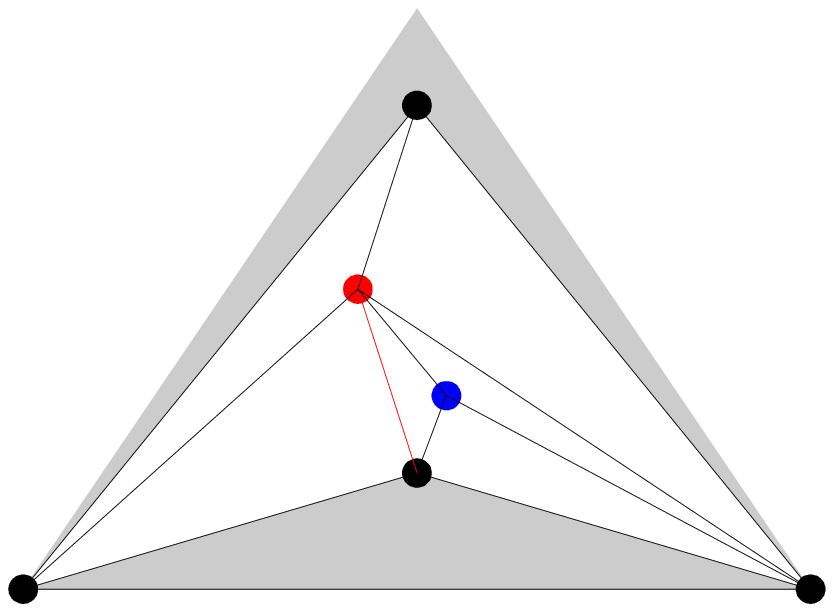,width=\x}}
\centerline{
\psfig{file=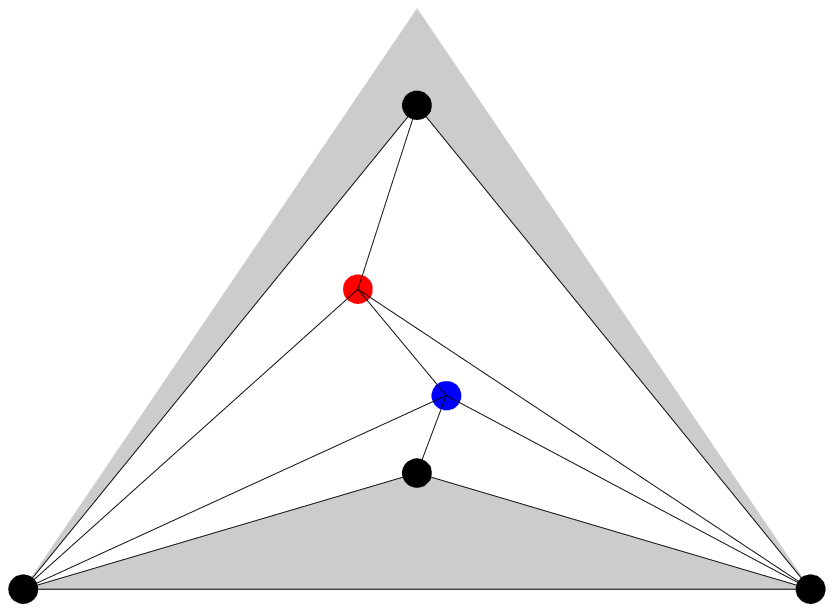,width=\x}\y 
\psfig{file=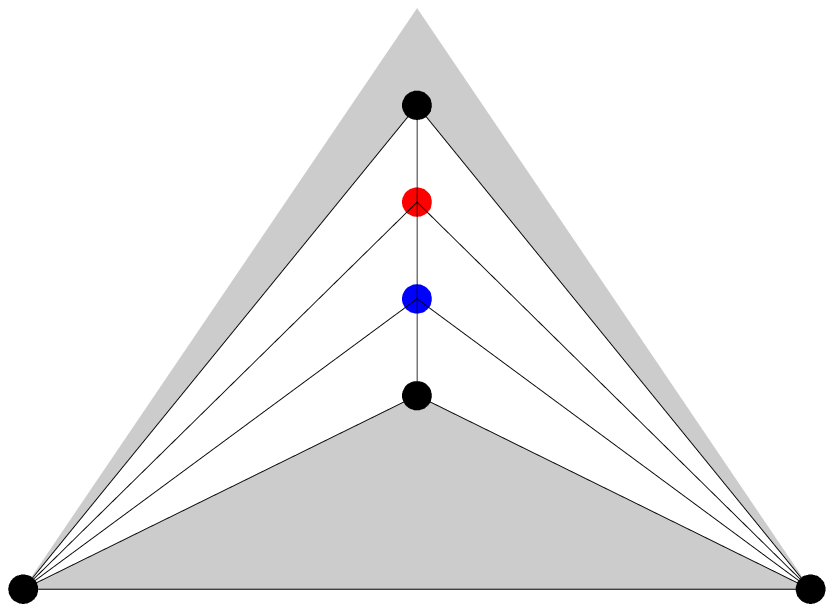,width=\x}\y 
\hphantom{\psfig{file=figs/bluered1.ps,width=\x}}
}
\caption{Exchanging blue and red in the stem of the christmas tree by
a sequence of moves. Top to bottom. The actual flips are taking place
between frames 1-2, 2-3, 5-6, 6-7. The other transitions are again
just moves for better visibility. Note that the shaded regions may
contain arbitrary links, in particular, parts of the stem of the
tree. Therefore, this sequence shows that any two positions in the
vertical stem can be exchanged, except for the bottom 2. Those will be
handled in \fref{bluered2}.} 
\label{bluered1}

\end{figure}
\begin{figure}[htb]
  \centerline{
\psfig{file=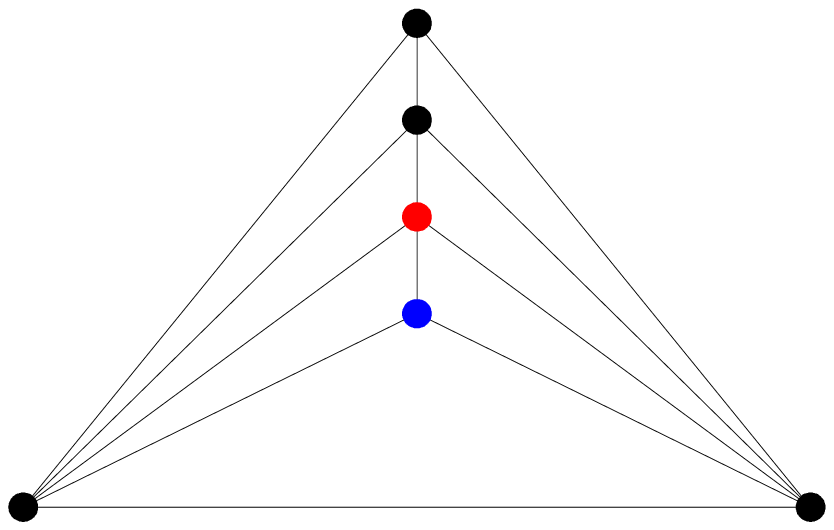,width=\x}\y
\psfig{file=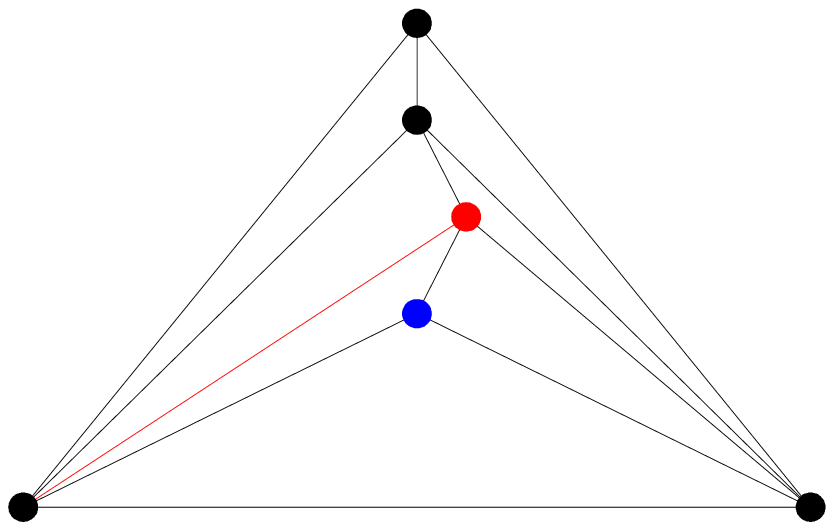,width=\x}\y 
\psfig{file=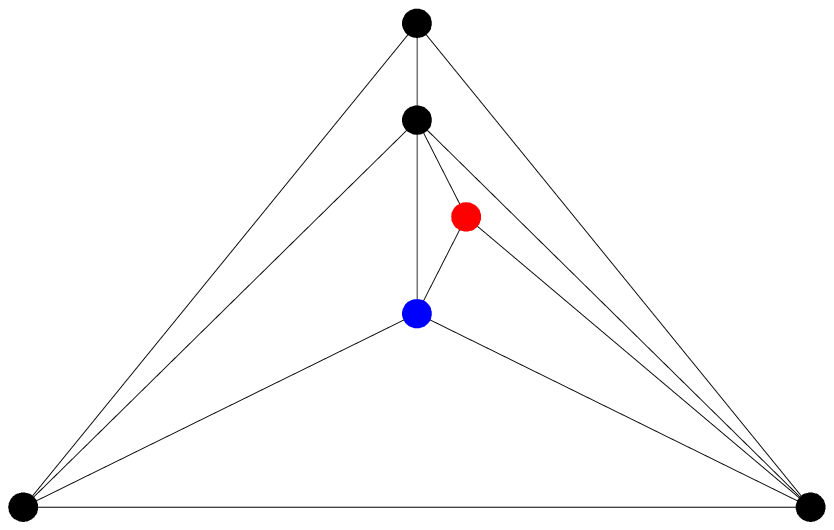,width=\x}}
\centerline{
\psfig{file=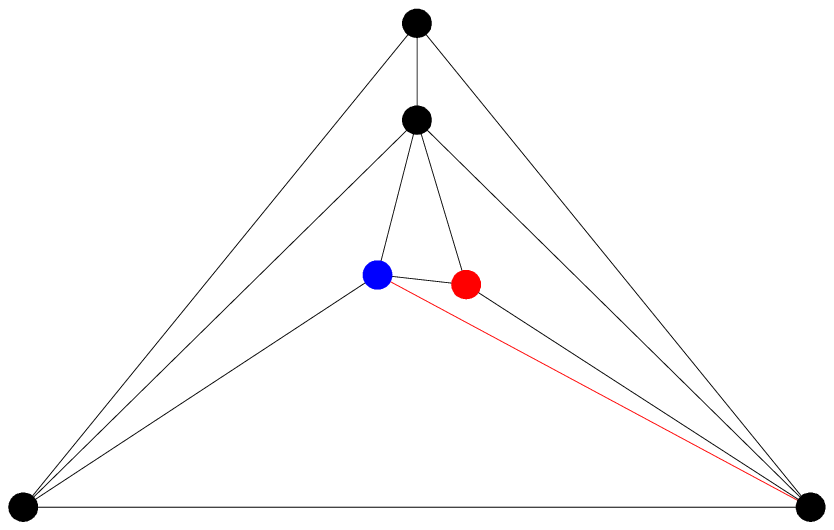,width=\x}\y 
\psfig{file=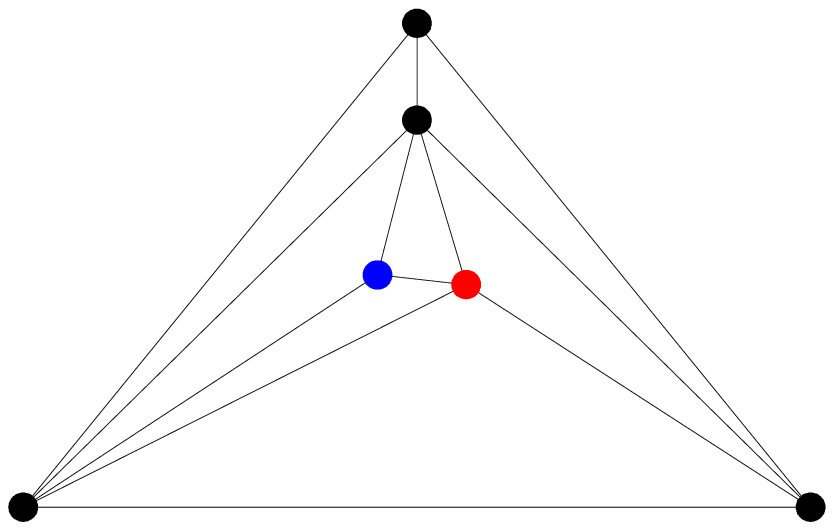,width=\x}\y 
\psfig{file=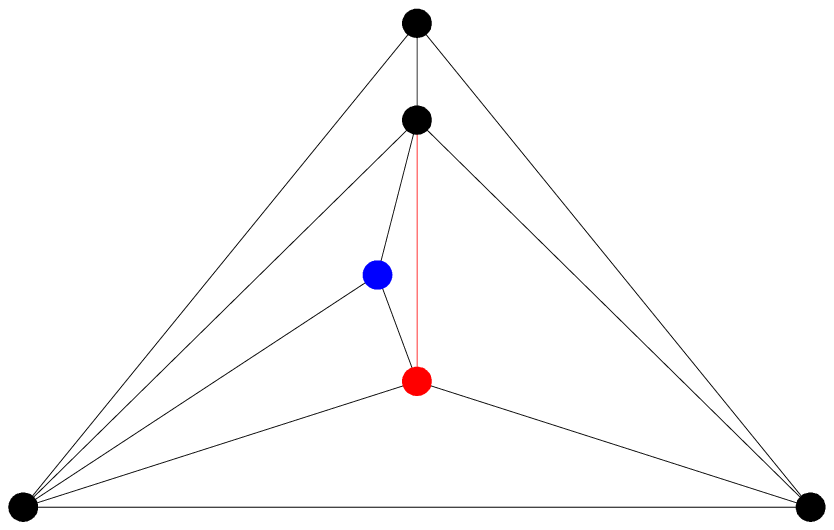,width=\x}}
\centerline{
\psfig{file=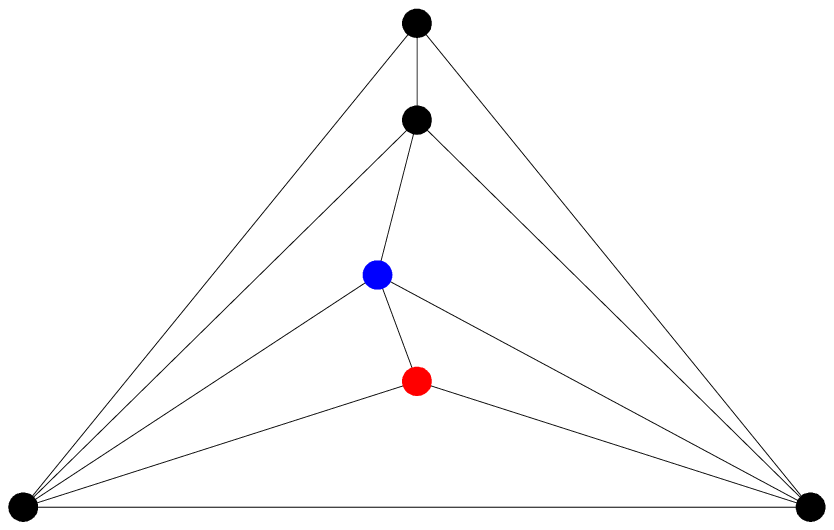,width=\x}\y 
\hphantom{\psfig{file=figs/blueredbottom7.ps,width=\x}}\y 
\hphantom{\psfig{file=figs/blueredbottom1.ps,width=\x}}
}
\caption{Exchanging blue and red in the two bottom positions of the
  christmas stem. Top to bottom. The flips occur between frames 2-3,
  4-5, and 6-7.}
\label{bluered2}

\end{figure}

\begin{remark}\label{r:tree}
Note that the energy of the christmas tree (with $n$ nodes) is at
least (depending on the distribution of colors)
\begin{equ}
E_{\rm tree}(n)\ge  2\cdot(7-(n-1))^2+ (n-3)\cdot(5-4)^2
+1\cdot(5-3)^2=2n^2-31n+129~,
\end{equ}
when $n\ge8$. It is important to note  that the christmas tree is an
``expensive'' configuration energy-wise, but very convenient as a
topological anchor from which to reach other configurations.
\end{remark}

\sect{The phase space}{phasespace}

In this section, we describe the phase space $\TT_n$ of the system. The possible
states of our system of triangulations with $n$ nodes is the set $\TT_{n}$
of all possible
colored triangulations. The set $\TT_n$ has, as we will see, a number of elements
which grows like $C^n$ for some constant $C$. It is thus a discrete
space with a finite number of states. To describe the dynamics of
flipping in a geometric way, one should view this set as the
\emph{dynamical graph} 
$\GG$, in which the nodes are the elements of the set $\TT_{n}$ and two
nodes are linked if one can be reached from the other by a
flip. (This makes an undirected graph, since one can flip back and
forth.) The reader should note that there are two graphs in this
discussion: Each triangulation is a graph with $n$ nodes, and $3n-6$
links (by Euler's theorem), while the graph $\GG$ has about $C^n$
nodes, and about $3n-6$ links \emph{per node}. This last statement
follows because in every state of $\TT_{n}$, one can choose which of the
$3n-6$ links of the triangulation $\T$ one wants to flip. However,
there will, in general, be somewhat fewer links which are candidates
for flipping, because whenever there is a node of degree $3$ in the
triangulation $\T$ its links can not be flipped (a tetrahedron is
unflippable).

In more physical language, comparing with the local degree in $\ZZ^d$,
(which is $2d$), one can say that the ``local dimension'' of the
dynamical graph
$\GG$ is something like $\OO(n)$, while the size of the phase
space (the number of nodes in $\GG$) is $\OO(1)^n$.

Finally, given any two elements in $\TT_{n}$, that is, any two
triangulations with $n$ nodes, we will show below that $\OO(n^2)$ flips
are sufficient to walk on the graph $\GG$ from one to the other. Thus,
 the diameter of the graph $\GG$ is at most $\OO(n^2)$ while it has
$\OO(1)^n$ vertices. This means that $\GG$ has the ``small world'' property
 \cite{Watts1998}. It has also small clustering coefficient, since
 there are very few triangles in the graph $\GG$ (it is difficult to
 get form a triangulation back to the same triangulation with 3 flips).

In the remainder of this section, we prove these statements. They are
well-known for uncolored graphs, so the only task is to prove them for
the colored graphs. 

We first state two known results for the set $\TT_{n,0}$ of uncolored
triangulations: 
\begin{lemma}\label{l:tutte}\cite{Tutte1962,Negami1999,Mori2003}
  The number of elements in $\TT_{n,0}$ is asymptotically
  \begin{equ}\label{tutte}
\left(\frac{256}{27}\right)^{n-3} \frac{3}{16 \sqrt{6\pi n^5}}~.    
  \end{equ}
The distance between any two uncolored triangulations is 
at most $6n-30$ flips.
\end{lemma}
For the case of the colored graphs, with $n_{\rm red} = n_{\rm blue}
+c$
and $c\in\{0,1\}$, that is, about equal number of red and blue nodes,
one has
\begin{lemma}\label{l:tuttecolor}
  The number of elements in $\TT_n$ is asymptotically bounded above by
  \begin{equ}\label{tuttecolor}
2^n \left(\frac{256}{27}\right)^{n-3} \frac{3}{16 \sqrt{6\pi n^5}}~,
  \end{equ}
and below by the expression \eref{tutte}. 
The distance between any two colored triangulations in $\TT_n$ is bounded by
\begin{equ}[diam]
  C_1 n^2 + C_2
\end{equ}
flips with some universal constants $C_1$, $C_2$.
\end{lemma}

\begin{remark}
  This result might be compared to glass models on cubes. In that
  case, one has also the small world property
  \cite{benarous2006}. However, our model is not ``trap''-like, since
  there are no very deep holes but rather very narrow corridors, see
  also \sref{s:nature}.
\end{remark}
\begin{proof} The lower bound in \eref{tuttecolor} is obvious from Lemma~\ref{l:tutte},
  since there are certainly more triangulations with coloring than
  without. The upper bound follows by observing that there cannot be
  more than $2^n$  different colorings of the nodes of any uncolored
  triangulation. To estimate the number of steps needed to connect two
  colored triangulations, we reduce the problem to the uncolored
  one. Starting from an arbitrary triangulation $\T$ in $\TT_n$, we can
  go to $\T'\in\TT_n$ \emph{without respecting the colors}, by ``passing
  through the christmas tree''. This needs at most $12n-60$
  flips. However, the colors in the final position might be wrong and
  they must be reordered. We do this not at the end, but when we are at
  the christmas tree. Here we use the method described in
  \fref{bluered1} and \ref{bluered2}. Each permutation of two
  neighboring colors can be done by at most 4 flips. Since no color
  has to be moved by more than $n-1$ positions, and there are $n$
  nodes we get the bound \eref{diam} as asserted.
\qed
\end{proof}

\sect{The energies}{s:cost}

Up to now, our discussion has been purely topological. But there is
also energy. The shortest paths (of length at most $C_1n^2 +C_2$
as described in \lref{l:tuttecolor}) to go from $\T$ to $\T'$ are by
no means energetically optimal, and optimal paths are difficult to
find. We have already seen in Remark~\ref{r:tree} that the christmas
tree has energy $\OO(n^2)$. The minimal energy of the model is
clearly 0, by \eref{e:energy}. However, this energy can not be quite
reached, because of Euler's theorem. We have the following,
probably non-optimal result:
\begin{lemma}\label{l:minenergy}
  For every $n=18+12k$ with $k\in\natural$, there is a triangulation
  $\T$ in $\TT_n$ (with an equal number of red and blue nodes) whose
  energy $E(\T)$ is between 6 and 54.
\end{lemma}
\begin{corollary}\label{c:minenergy}
  There is a constant $C$ such that for every $n$ there is a
  triangulation $\T\in\TT_n$ (with the number of red and  blue nodes
  differing by at most 1) such that $E(\T)\le C$.
\end{corollary}
\begin{figure}[!htb]
   \centerline{\psfig{file=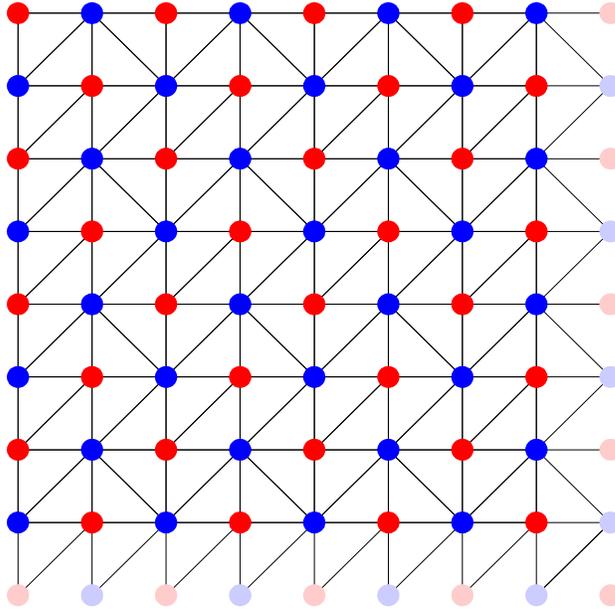,width=8.2cm}}
\caption{A triangulation of the torus with 64 nodes and energy 0. We
  show the 64 nodes and the periodic extension (softer colors).}
 \label{f:torus}
\end{figure}
\begin{remark}
  Note that the statements above are by and large independent of the
  choice of the function $E$, provided it has its local minima at
  $5$ and $7$. Of course, the constants will depend on the details of
  $E$ but the basic facts will not. What will change, however, is, \eg,
  the highest possible energy. In our case, it is $\OO(n^2)$ but for
  potentials which grow faster it will be higher.
\end{remark}

\begin{remark}\label{r:torus}
  One can also study triangulations of the torus (where the Euler
  characteristics is 0). In that case, it is
  easy to see, and in fact shown in Fig.~3 of \cite{Proglass2007},
  that there is a state of energy 0 when $n$ is a multiple of 4. This
  state is a regular arrangement of nodes of degree 5 and 7 (see
  \fref{f:torus}). Note that flipping all the links connecting blue
  nodes in the 2nd column will exchange blue and red there and will
  generate another configuration of energy 0. Since this can be done
  for all even columns independently, the degeneracy of the ground state of torus
  triangulations with $n=4k$ nodes is at least about $2^{\sqrt{k}}$,
  which makes 
  it quite degenerate. One can play the same game with horizontal
  rows, but this does not change the square root behavior of the exponent.
\end{remark}

\begin{proof}[of Lemma \ref{l:minenergy}]
  The proof is by construction. In \fref{f:5742} we show a
  triangulation which has 18 nodes, and energy $E$ between 6 (which seems to
  be the minimal possible energy) and 54. Note that the shaded triangle has
  the same number of internal links from its corners than the
  outermost triangle (namely 6). Therefore, we can repeat the
  construction recursively in the interior triangle by adding another 12 nodes (6
  blue and 6 red), \ie, the shaded triangle will look like the
  original one. Its 3 black nodes will become red. Therefore, the number of black dots will not
  increase, and we see that for $k\in\natural$ there is a
  triangulation with $18+12k$ nodes, with energy between 6 and 54, as
  asserted.
The corollary follows immediately: If $n=18+12k+\ell$ with
  $0\le \ell\le 11$, we just do the construction for $n'=18+12k$ nodes
  and add the additional $\ell$ nodes inside the innermost (black) triangle,
  and connecting them to make a triangulation. This subgraph and its
  connections to the black nodes will increase the energy by some
  finite, $k$-independent amount.\qed
\end{proof}
\begin{remark}
  The bounds of \lref{l:minenergy} are not optimal. For better values,
  see also the numerical  studies of \sref{s:numerics}.
\end{remark}
\begin{figure}[!htb]
  \centerline{\psfig{file=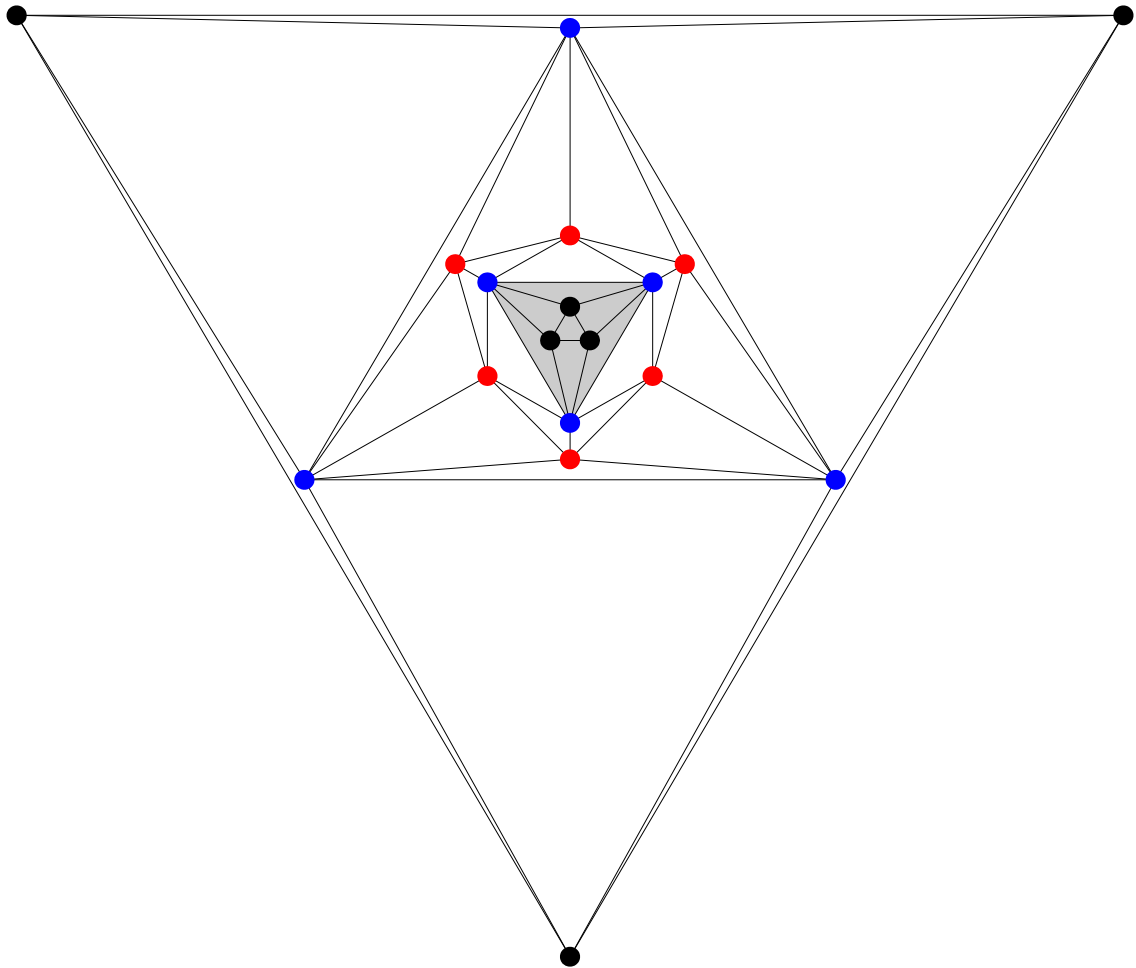,width=\textwidth}}
\caption{A triangulation $\T$ with 18 nodes and energy $E(\T)\le 54$. The
  blue nodes have degree 7 and the red ones have degree 5. The black
  ones have degree 4, which means that the energy of the triangulation
  is at least $6\cdot(5-4)^2=6$ and at most $6\cdot(7-4)^2=54$.}
\label{f:5742}
\end{figure}

We next give a bound on the degeneracy of the energy levels. Given $n$,
one can ask about the number $N(n,E)$ of triangulations of $\TT_n$ of energy $\le
E$,
and of course, one can ask about their distribution in the limit
$n\to\infty$. Here, we only have a lower bound on $N$, which is certainly not optimal. But this bound will show that
for intermediate energies $N(n,E)$ grows at least like $C(E)^n$
as $n\to\infty$. 
\begin{lemma}\label{l:states}
There are constants $C_*>1$ and  $E_*>0$ such that the following holds.
For every sufficiently large $n$ and every $100< E< E_* n$ one has the
lower bound
\begin{equ}\label{e:multiplicity}
  N(n,E)\ge C_*^E~.
\end{equ}
 \end{lemma}
\begin{remark}
  The inequality \eref{e:multiplicity} can also be interpreted by
  saying that the number of states with energy $\epsilon n$ and
  $\epsilon <E_*$ is at least
  \begin{equ}
    N(n,\epsilon n) \ge C_*^{\epsilon n}~.
  \end{equ}
This should be compared to the growth rate of the number of
triangulations in $\TT_n$ which is also of the form $\const^n$.
\end{remark}
\begin{proof}
We give again a constructive proof, with no attempt to optimality.
We begin with a graph $G_m$ of the type of \fref{f:5742} with $m=18+12k$ nodes.
In every triangle of $G_m$  we may insert a graph $H$ which we now describe.
The graph $H$ is of the following form: First extend the triangle
``inward'' as shown in the leftmost panel of \fref{f:rotate} (the
triangle with the red point). Then the inner gray triangle is filled with
the triangulation of \fref{f:5742}, \ie, with $G_1$. Finally in the innermost triangle
of \fref{f:5742} we draw a christmas tree by adding 2 vertices. This
tree will break the rotation symmetry of \fref{f:5742}.
The whole construction adds 20 nodes to the original graph $G_m$, for every
insertion of $H$
(3 nodes in the first step, another 15 to place
\fref{f:5742} and 2 more for the tree).
Since there are $2m-4$ triangles in the original graph
there can be at most $2m-5$ insertions of $H$, namely one per original
triangle (we do not insert into
the outside of the basic triangle).

Note now that rotating the interior graphs in $H$ as shown in the 3
panels of \fref{f:rotate} will create a different graph for each
rotation of each of the inserted $H$, except perhaps for one overall
rotation. \emph{All these graphs have the same energy}. Thus, if there
are $p$ insertions of $H$ there will be $3^p$ graphs with the same
energy and the same $n$.

We next give an upper bound on the change of energy when inserting $p$ of
these graphs $H$. Inserting a triangle as in the left panel of
\fref{f:rotate} will increase the degree at each of the three nodes by
1. Since the original degree was at most 7, the new degree is at most
14 and hence, for every inserted $H$ this contribution can raise the
energy by at most some $E_0  (= 3\cdot (14-5)^2)$. (In fact it will be
less than that if not all triangles are filled with $H$'s.)
The energy of each graph $H$ itself is bounded by some other constant
$E_1$ ($=54 $ plus the contribution from the inserted tree). Thus,
upon inserting $p$ graphs $H$ the energy will grow at most by $p E_2=
p\cdot (E_0+E_1)$. 

Summarizing, we see that when we start with $m$ nodes and insert $p$
graphs $H$ we get $3^p$ graphs of the same energy by rotating the $p$
insertions separately. Furthermore, the order of the graph is $m+20p$
and the energy is less than $54+p E_2$, where the 54 is the bound of
Lemma~\ref{l:minenergy}.

One can now rearrange this statement to obtain the claim of
Lemma~\ref{l:states}. 
Since we started with $m$ nodes and were able to insert at most $2m-5$
graphs $H$, each of which adds 20 nodes, we get the inequalities
  \begin{equ}
    p\le 2m -5~,\quad n=m+20 p, \quad E \le 54 + E_2 p,\quad N=3^p~.
  \end{equ}
From this we conclude that the inequalities are satisfied if
\begin{equ}
  (p+5)/2 \le m \le n -20 p~,
\end{equ}
that is if
$  p\le \frac{2n-5}{41}
$
which for large $n$ is satisfied for $p\le n/25$. This means that the
number of insertions for which our construction works is bounded by
$n/25$ and the corresponding energy of such graphs is bounded by 
$54+ E_2 p \le 54+ E_2 n/25 \le E_* n$, with $E_*$ for example equal
to $E_2/50$ when $n$ is large enough. And in all these cases we
have $3^p$ graphs with the same energy, that is, at least
$3^{(E-54)/E_2}$. Introducing the lower bound $E>100$, we get
\eref{e:multiplicity} and the proof is complete.\qed
\end{proof}
\begin{remark}
  One can obtain somewhat less good bounds by inserting directly
  christmas trees into each triangle of $G_m$.
\end{remark}
\begin{figure}[!htb]
\def\x{4.5cm}
\def\y{\kern2em}
\centerline{
\psfig{file=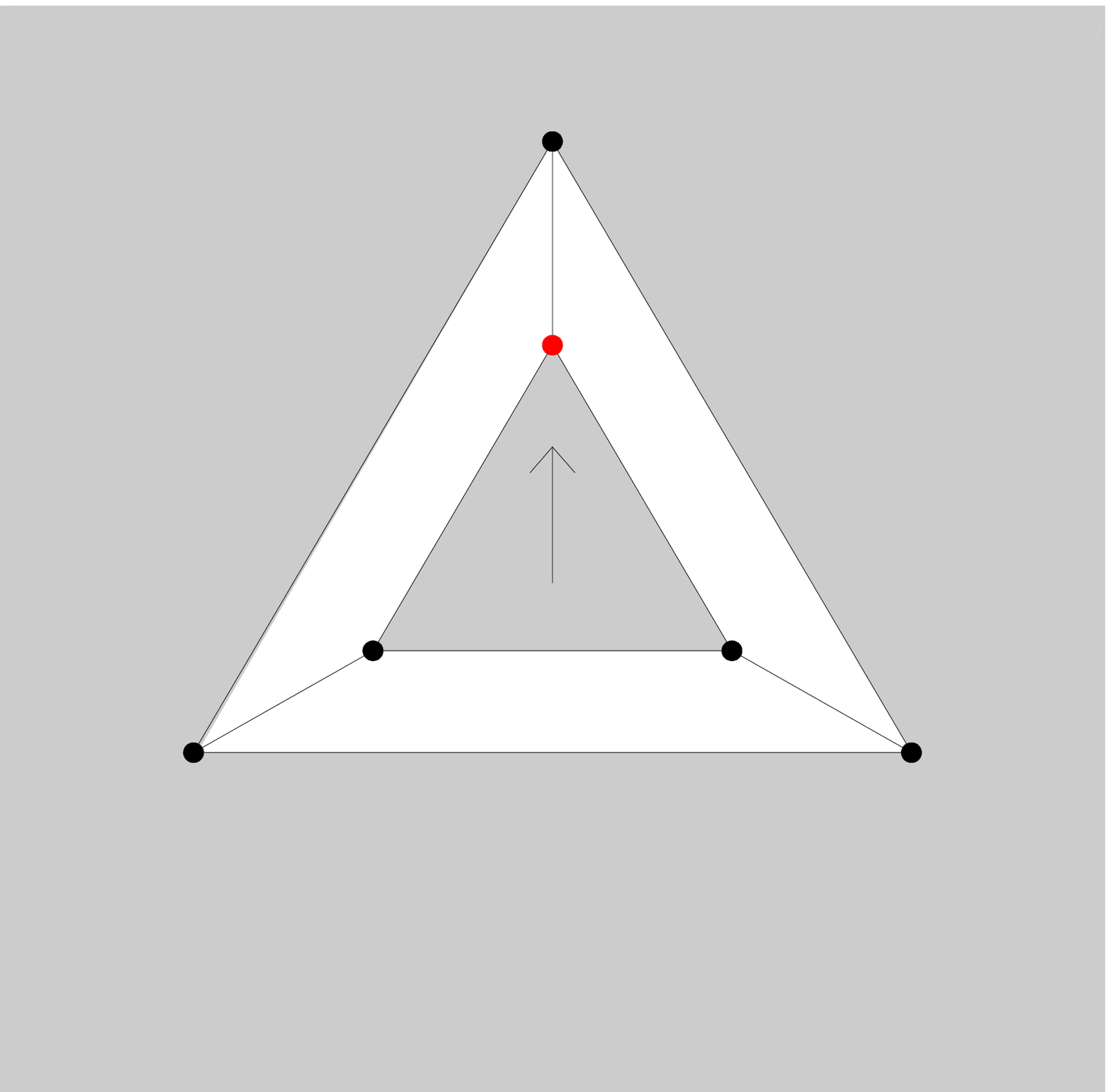,width=\x}
\y
\psfig{file=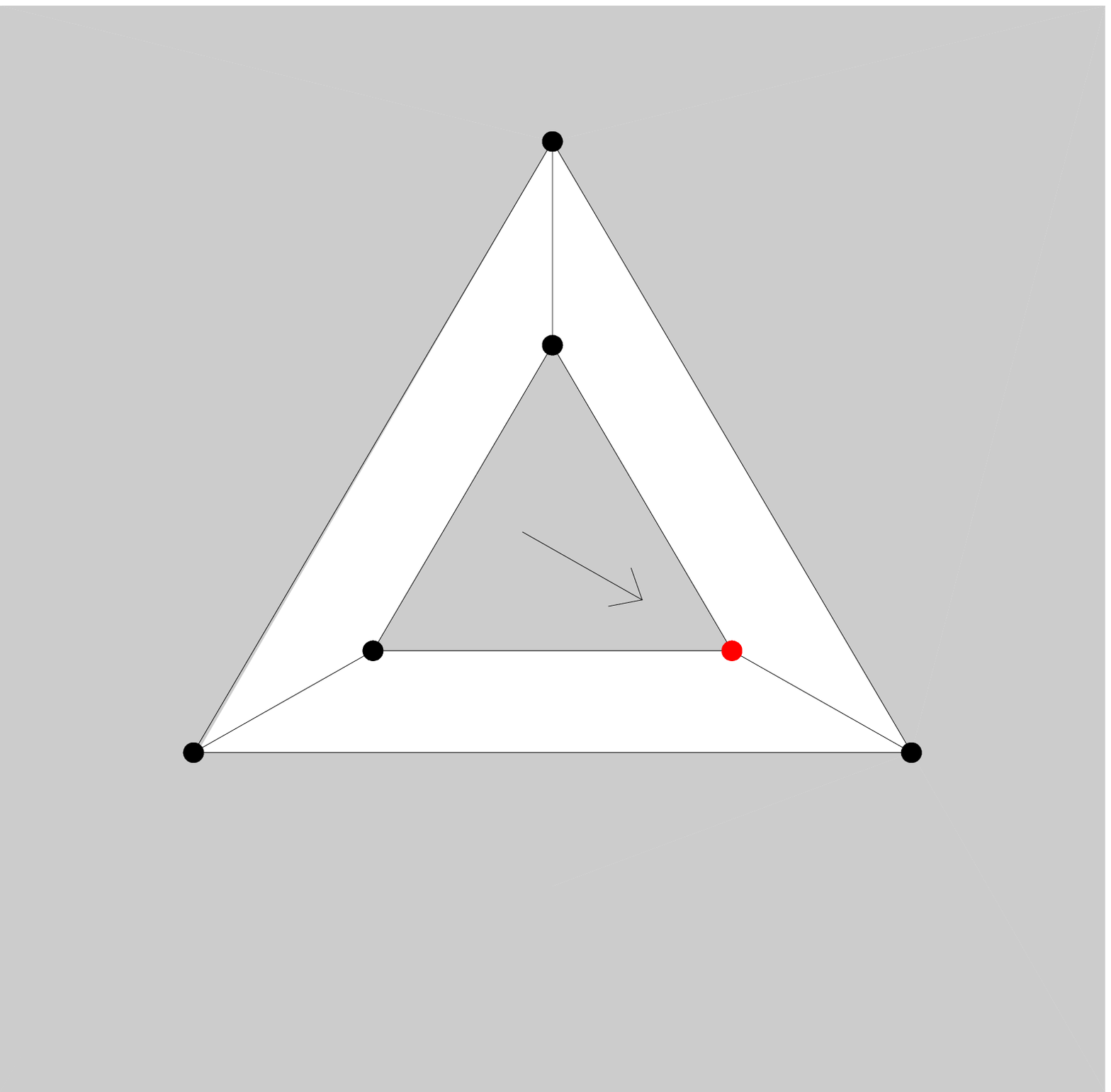,width=\x}
\y
\psfig{file=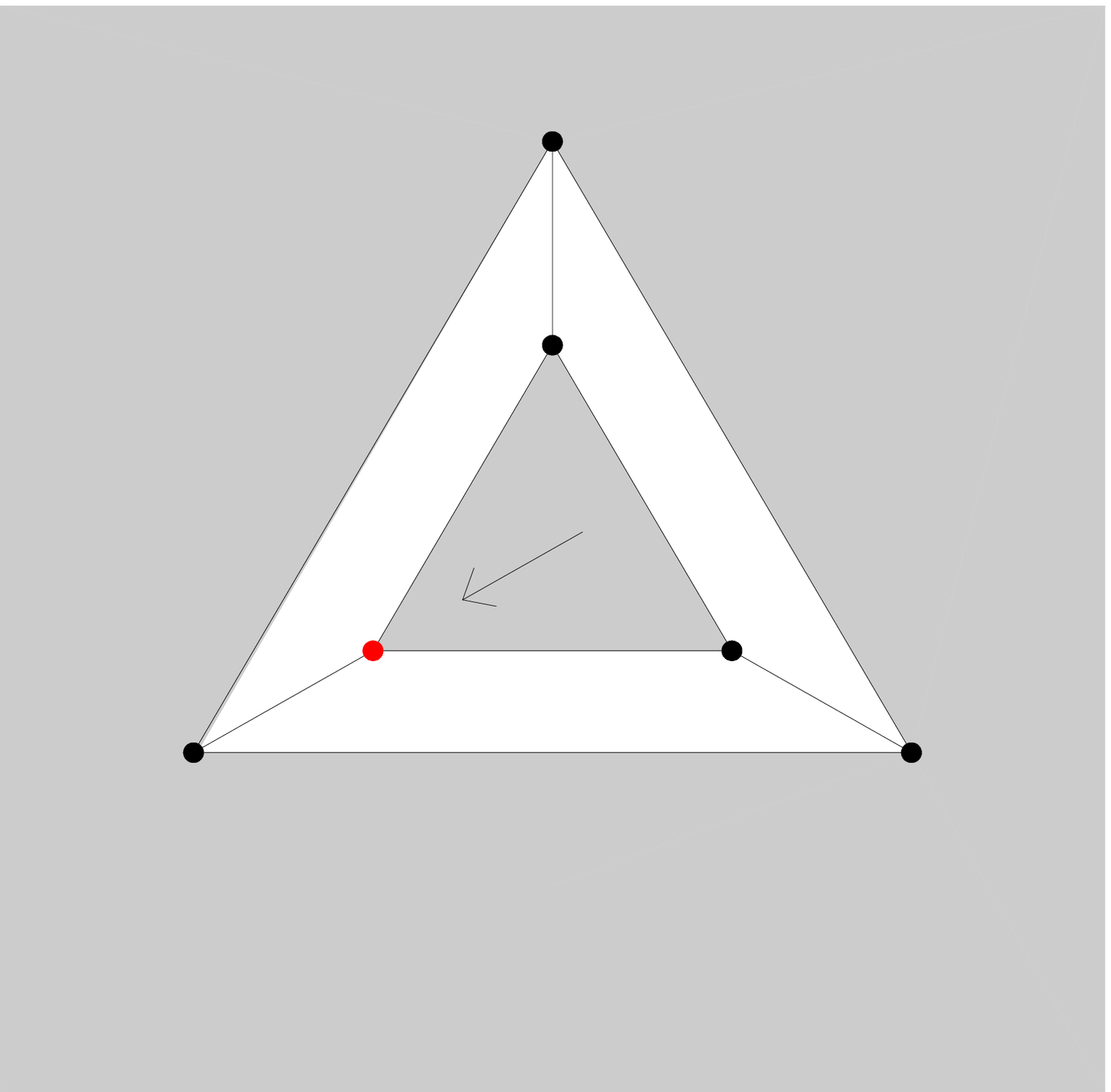,width=\x}
}
\caption{Three configurations, in the interior of an arbitrary
  triangle, which have the same energy (the inserted figure being one
  with three external nodes of same degree and same color). These three configurations are
  different unless the interior triangle has 3-fold symmetry.}
\label{f:rotate}  
\end{figure}

\sect{The dynamics of defects}{s:defects}

A useful way to view a typical state of low energy is the notion of
defects. Let us call \emph{defect} any red node whose degree is not 5 
and any blue node whose degree is not 7. Thus, the graph of \fref{f:5742} has 6
defects. When the energy of a triangulation of $n$ nodes is less than
$\epsilon n$---which is quite frequent when $n$ is large, as we saw in
Lemma~\ref{l:states}---then there are very few defects, since each
costs at least one unit of energy. Therefore, at these very low energies
it is useful to view the triangulation as a dilute gas of defects. 

Here we want to show that these defects can actually move. Again, this
is illustrated by an example. In \fref{f:movedefect}, we show a
sequence of 3 flips with the property that after the flips, one node
has degree lowered by 1 and a node at distance 2 has degree raised by
one, while all other nodes have the same degrees before and after the
flips. It is easy to see that this mechanism allows one to ``move'' a
defect by 2 steps with a small number of flips. Of course, by
Lemma~\ref{l:irreducibility} and Lemma~\ref{l:tuttecolor},
we already know that this can be done in $\OO(n^2)$ steps. But what is
new here is that the number of steps needed to move the defect by a
distance 2 is independent of the size of the triangulation.
\begin{figure}[!htb]
\def\x{4.5cm}
\def\y{\kern2em}
  \centerline{
\psfig{file=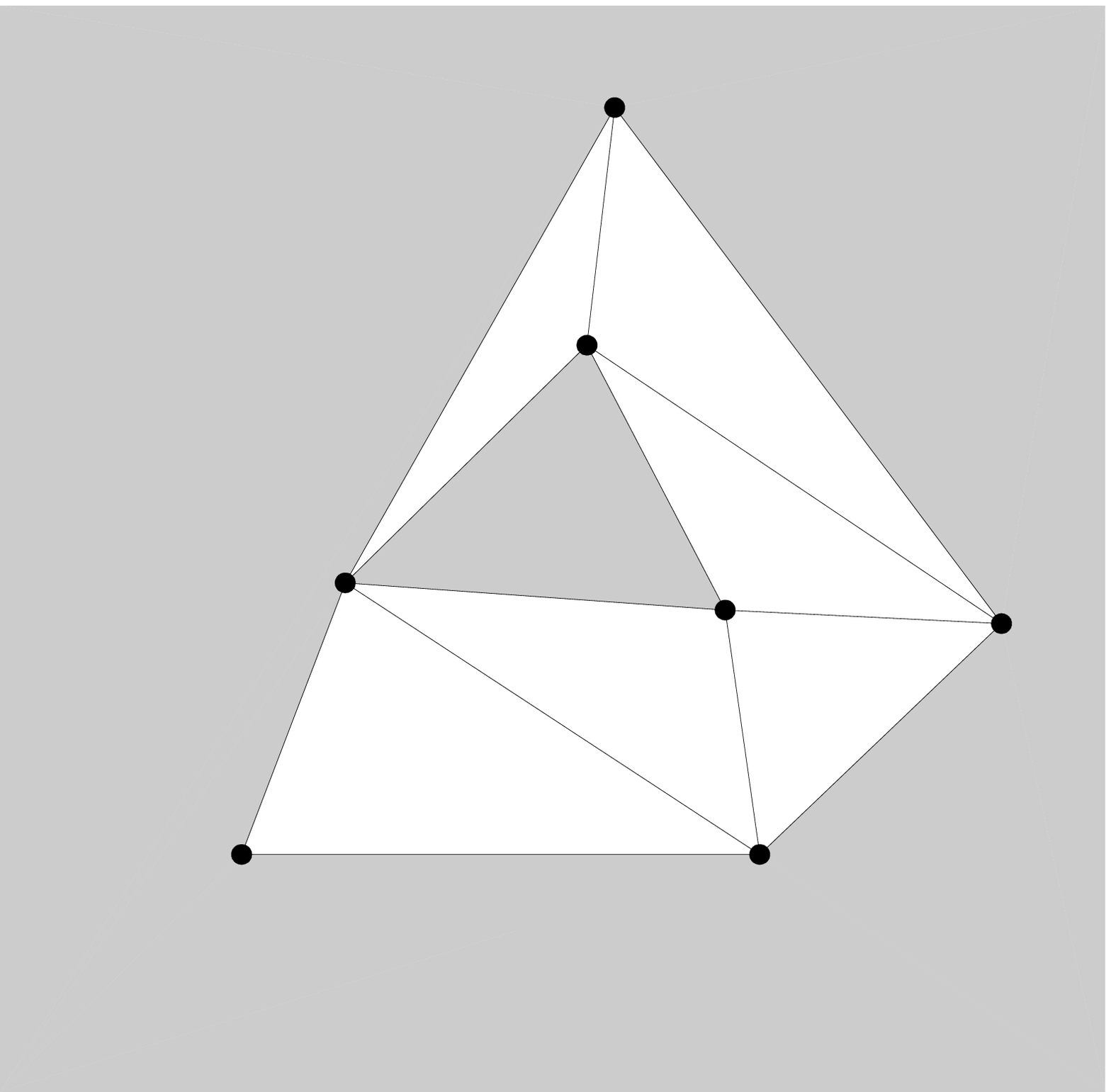,width=\x}\y
\psfig{file=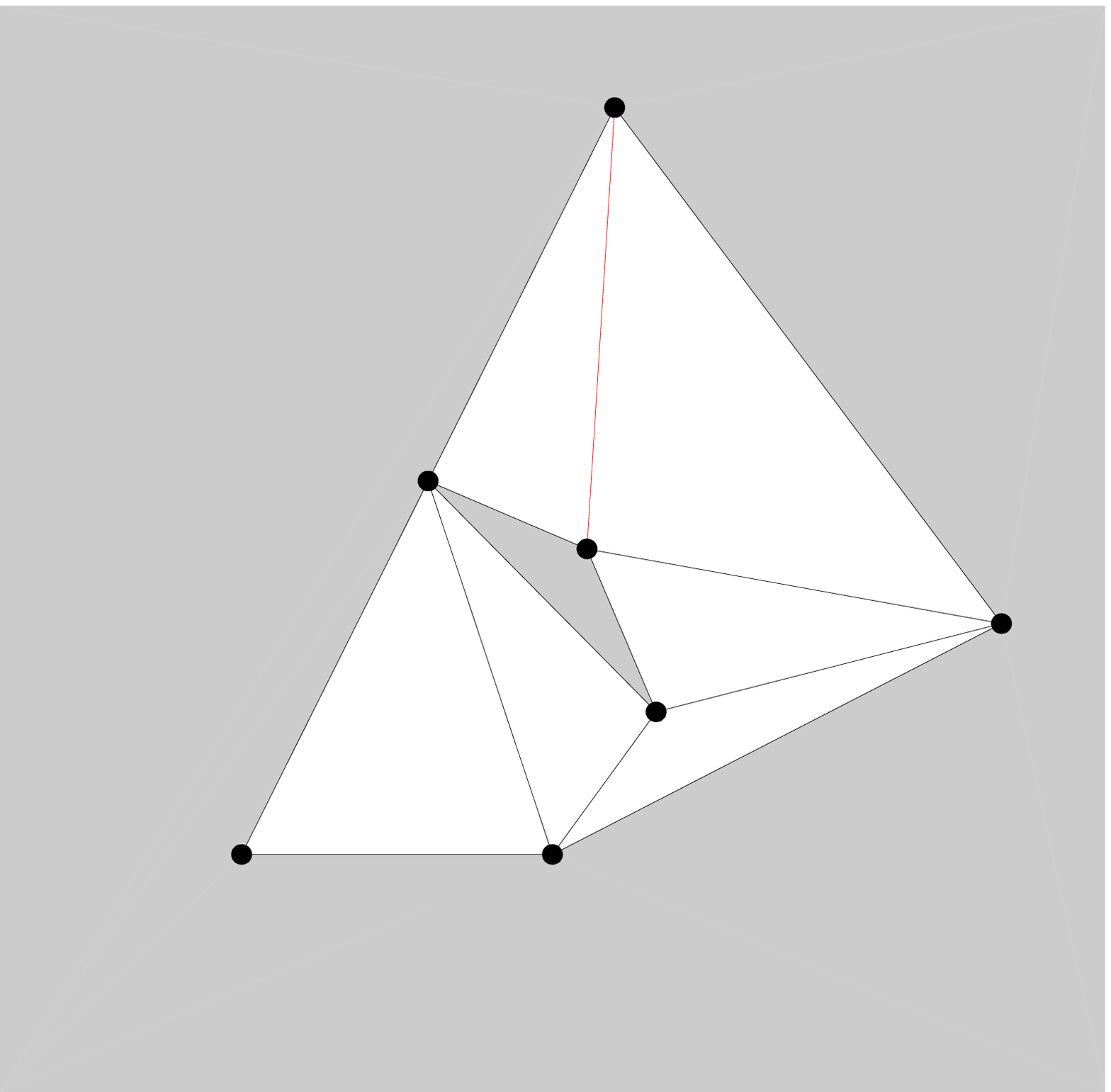,width=\x}\y
\psfig{file=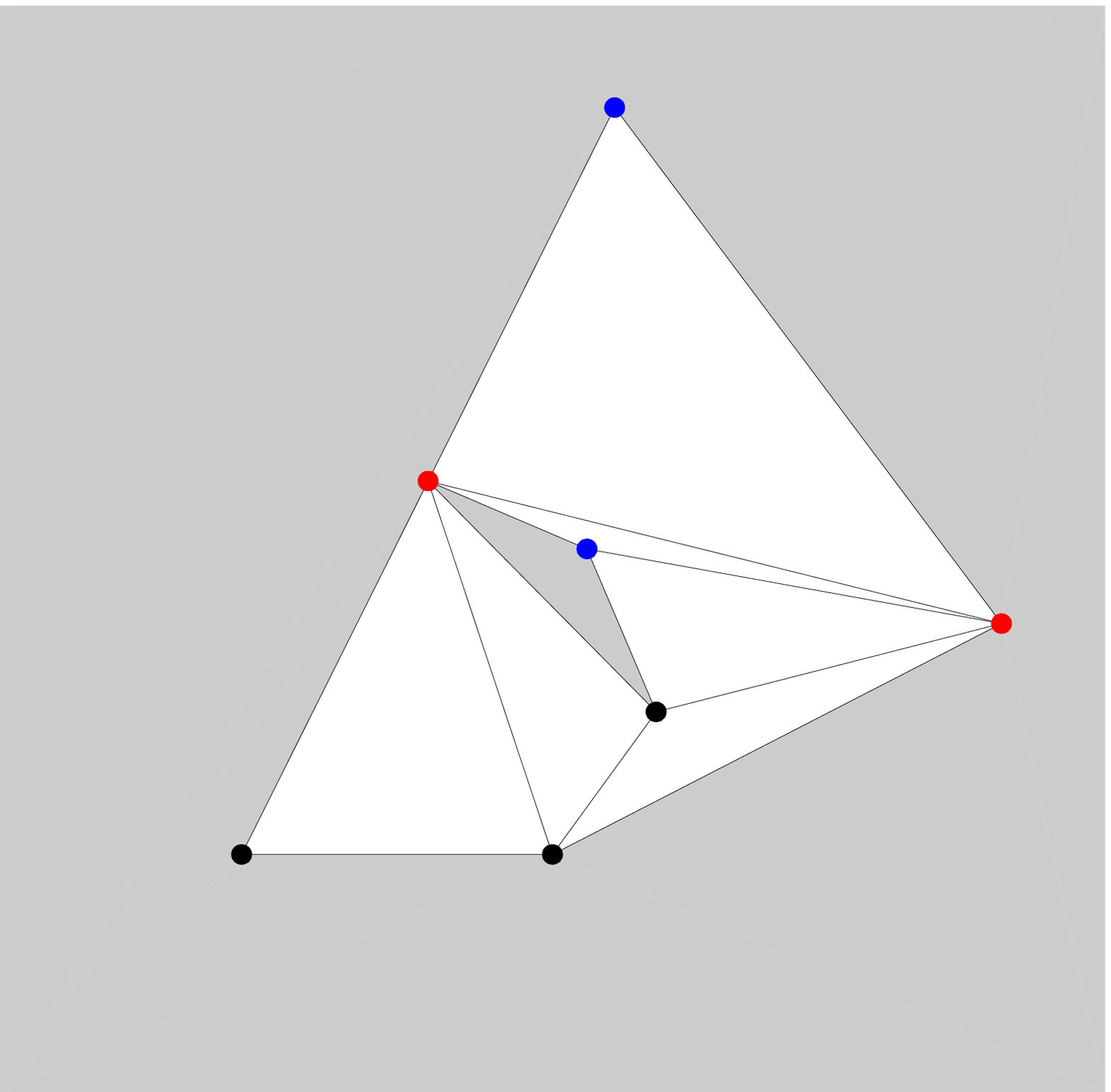,width=\x}
}
\vskip0.5cm
  \centerline{
\psfig{file=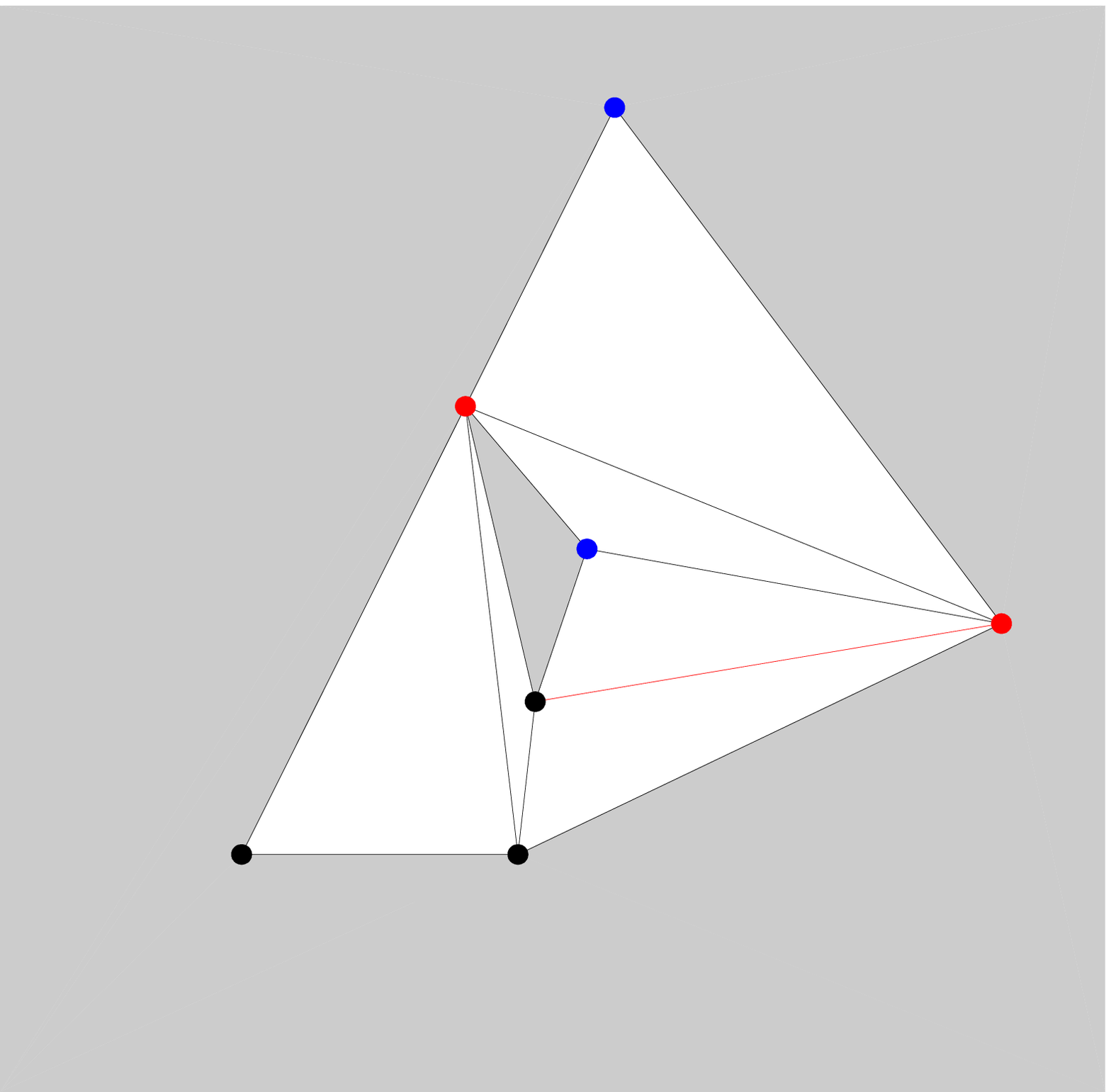,width=\x}\y
\psfig{file=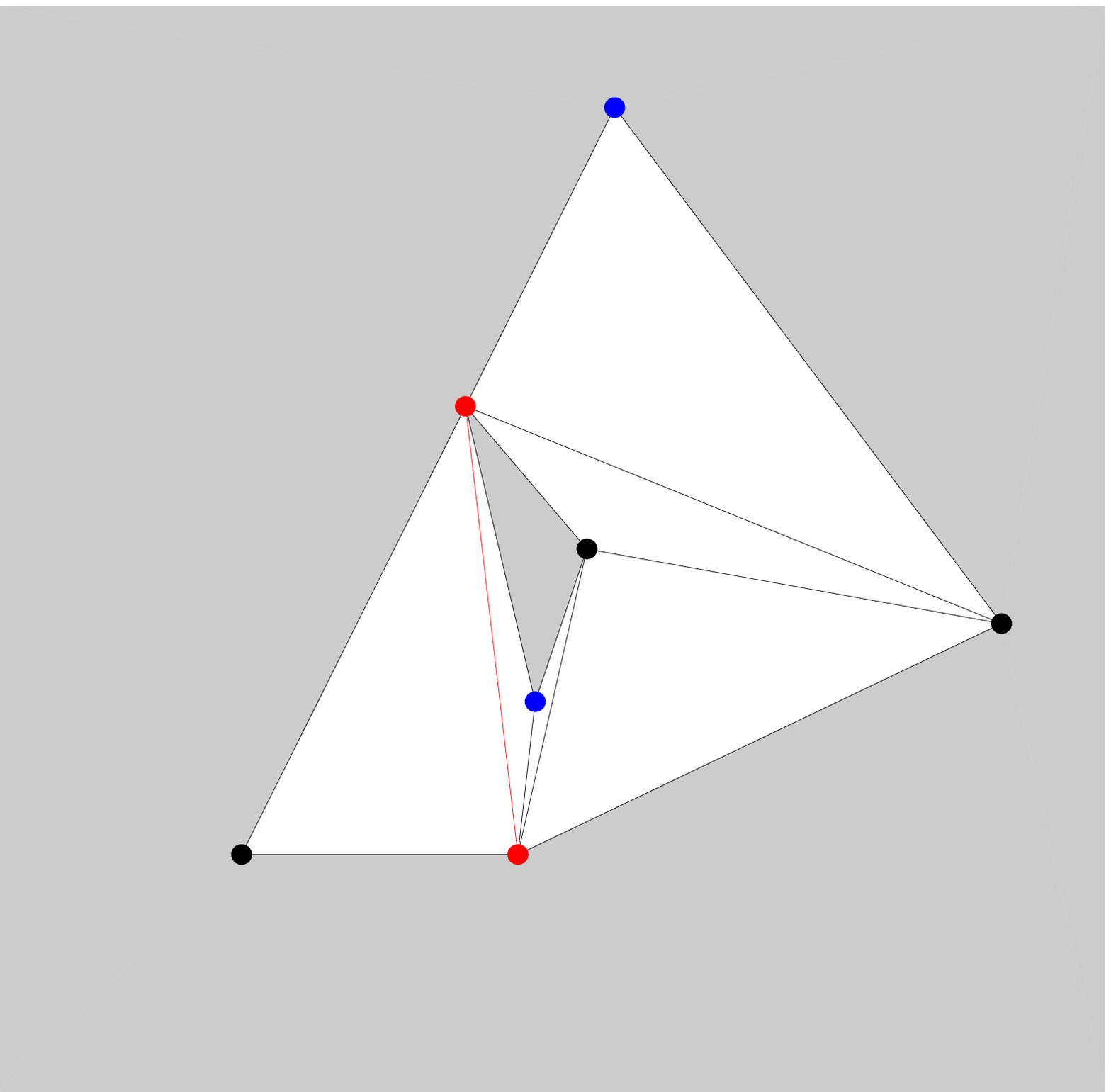,width=\x}\y
\psfig{file=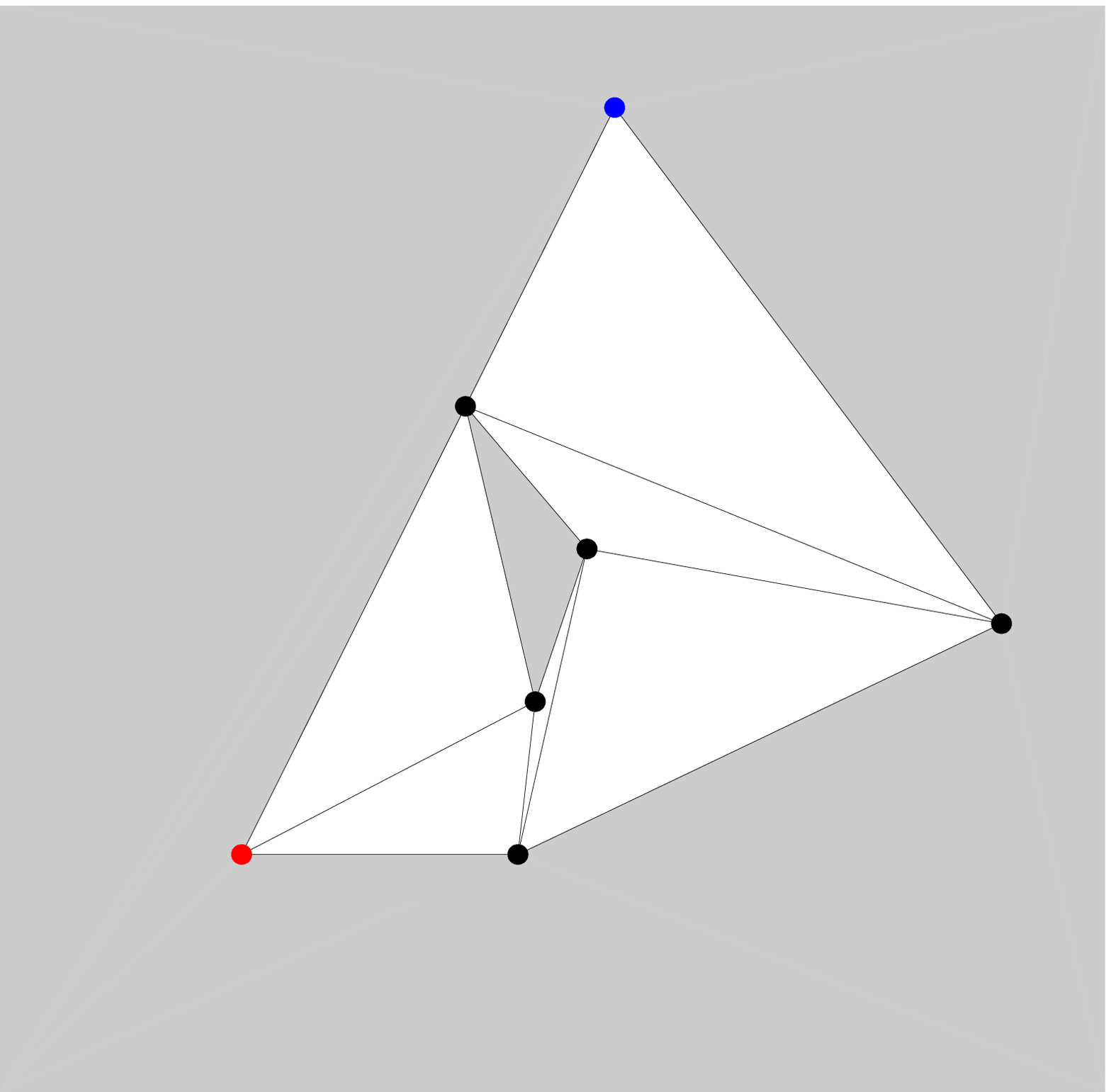,width=\x}
}
\caption{Motion of a defect by successive flips. Shown is a sequence of
  moves which lowers the degree by 1 on the top node and raises the
  degree by one on the bottom left node, leaving in the end all other
  degrees intact. Flips take only place between frames 2-3, 4-5, and
  5-6. The color code of the nodes is red for an increase of the degree
  by 1 and blue for decrease by one.}
\label{f:movedefect}
\end{figure}

Note that since defects can move, they can actually collide and
annihilate each other. In \sref{s:nature}, we will argue how this fits
nicely into a random walk picture, of a gas of annihilating defects 
combined with a rate of spontaneous generation of new defects.

\sect{Numerical results}{s:numerics}

\subsect{Time evolution of energy}{ss:energy}
We ran several simulations on triangulations of size $n=367$,
$1096$ and $3283$. The energy was defined as in
Eq.~\eref{e:energy}. All runs have been done starting from a fixed
initial configuration, with fixed temperature. 
The original triangulation is obtained by
  starting from a tetrahedron which is recursively subdivided by the
  insertion of tetrahedra. Thus, at level $\ell$ of recursion there
  are $4+\sum_{i=1}^\ell 3^i$ nodes. Our runs are for levels 5 to 7. 

In \fref{f:energy} we illustrate the typical scenario for the
evolution of the energy as a function of the number of flips, at
temperature $T=0.175$. After a short initial phase, there is a marked
decrease of the energy until it reaches an order of about 100 (for the
sizes of our triangulations).
Then a slow decay sets in until an equilibrium value is reached.

The absolute time scales seem proportional to $n$. So we can produce a
first data collapse by rescaling time by a factor $3283/n$. This is
done in \fref{f:all}.
\def\x{8cm}
\begin{figure}[!htb]
  \centerline{\psfig{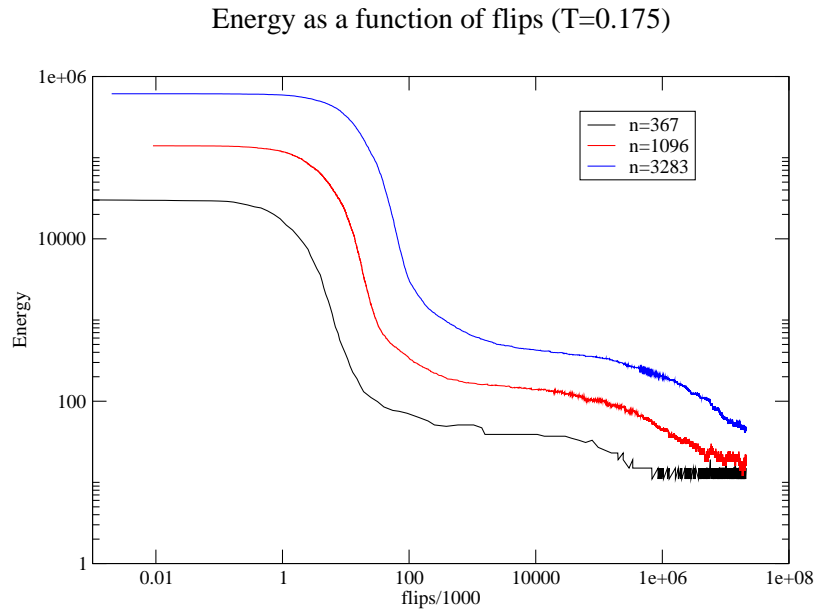}}
\caption{The energy as a function of the number of flips. The vertical
  scale is energy.}
\label{f:energy}
\end{figure}
\begin{figure}[!htb]
 \centerline{\psfig{file=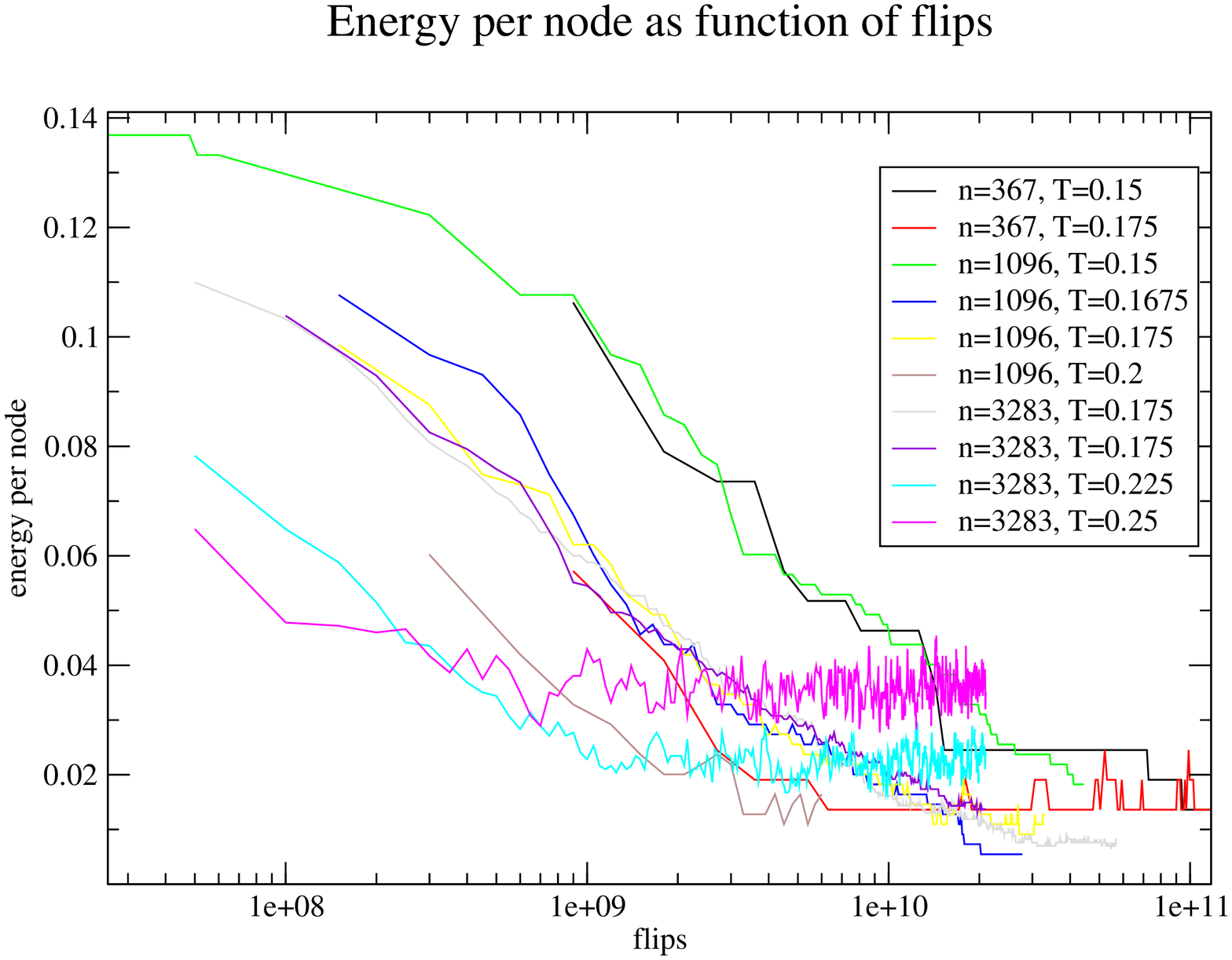,width=\x,angle=270}}
\caption{The data for $n=367$, $1096$, and $3283$ and $1/\beta=T$ as
  shown in legend.
  The $x$ axis is the logarithm of number of attempted flips,
  normalized by $3283/n$, and the $y$ axis energy per node after
  having subtracted a zero-point energy $E_0=6$ from all the
  energies. Note a certain data collapse for each temperature $T$,
  independent of size $n$.}
\label{f:all}
\end{figure}

Glassy slowing down is demonstrated in \fref{f:energyend}.
Here we concentrate on the phase just before equilibrium is
reached. With very good quality, one finds an exponential slowing down.
The energy behaves like the logarithm of the number of flips.
More precisely, before saturation, we find 
laws of the form (between $3\cdot 10^8$ and $3\cdot10^{10}$ flips):
\begin{equa}[e:scaling]
  (E_{367}-6)/367&= 0.44958-0.01897\log(3283/367\cdot{\text{flips}})~,\\
  (E_{1096}-6)/1096&=0.46169 -0.019362 \log(3283/1096\cdot{\text{flips}})~,\\
  (E_{3283}-6)/3283&=0.44554 -0.18652  \log({\text{flips}})~.\\
\end{equa}
The correlation coefficients of these fits increase from 0.977 to 0.997.
\begin{figure}[!htb]
  \centerline{\psfig{file=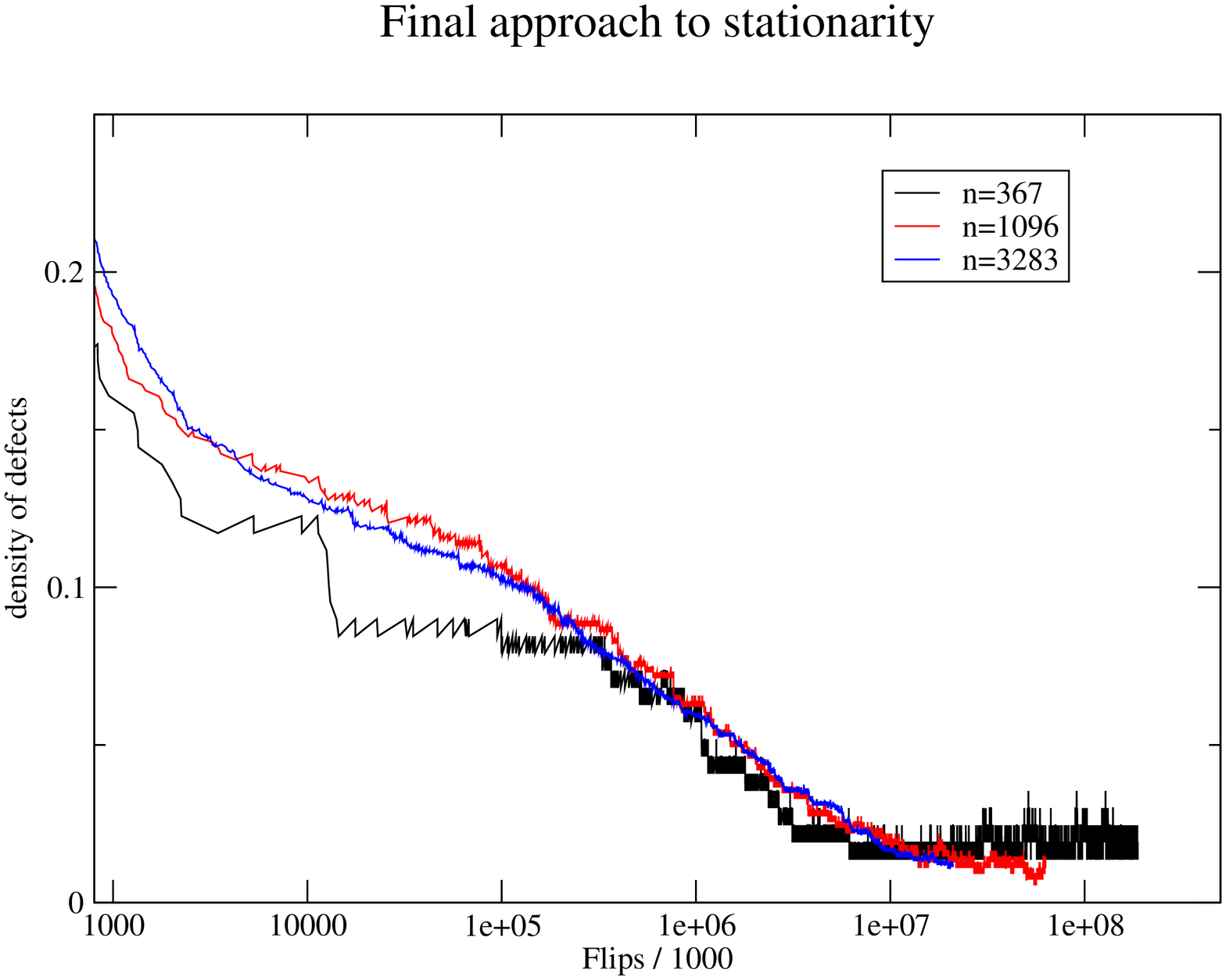,width=\x,angle=270}}
\caption{The energy as a function of the number of flips on approach
  to equilibrium. The vertical scale is density of defects, \ie,
  defects$/$nodes. The horizontal axis is the number of flips,
  renormalized by $3283/n$, \ie, the smaller triangulations converge
  faster than the larger ones. However, the \emph{slope} is the same
  for the 3 cases. We have subtracted 6 from the energies,
  to take into account the 0 point energy (which is at least 6).}
\label{f:energyend}
\end{figure}

\begin{remark}
  The minimal energies we have seen at $T=0.175$ are summarized in
  Table~\ref{tab:1}.
\end{remark}

\begin{table}
\small
$$  \begin{tabular}{|r|c|c|c|c|c|c|c|}
\hline
    $n$& 4 instead of 5 & 5 & 6 instead of 5 & 6 instead of 7 & 7 & 8
    instead of 7  & total energy  \\
\hline
    367&  6 & 178 & 0 & 5& 178 & 0 & 11\\
    1096&  10 & 536 & 2 & 4& 544 & 0& 16\\
    3283&  8 & 1629&  5 & 11& 1627 & 3& 27\\
\hline
  \end{tabular}
$$\caption{Some states with very low energy which have been found in the
  simulations. The numbers indicate the number of nodes of given
  degree. After the first column we give the counts for the red nodes
  (which want to have degree 5) and then for the blue nodes (which want
  to have degree 7).} 
\label{tab:1}
\end{table}

\subsect{Local minima and the ultrametric property}{ss:minima}
We have also studied the local neighborhood of typical triangulations
(at $T=0.175$). The results are shown in \fref{f:neighborhoods}. The
question we answer here is as follows. Take a ``typical'' low-energy
triangulation and compute for every possible flip starting from this
triangulation the change of energy $\delta E$ which that flip would
generate. For a triangulation with $n$ nodes there are in general
$3n-6$ possible 
choices of the edge which is going to be flipped, as discussed in
\sref{phasespace}. The figure shows the number of these flips which
change the energy by $\delta E= -2,0,\dots,10$, averaged over 200--700
states (depending on $n$). The values are expressed as probabilities.
In fact, the fluctuations between samples are very small and basically every single
sample has the same distribution. Also note that this distribution is largely
independent of the size of the system, except that very rare events
are absent in the smaller triangulations.

In terms of the energy landscape on the graph $\GG$ of
\sref{phasespace} this means that every point on $\GG$ which is a
typical glass state is almost a local minimum. In almost 100\% of all
directions (flips) leaving a given point, the energy grows by 4, with
fewer and fewer directions with different growth. Only about $0.0004$
of all directions are energy neutral, and the probability to find a
direction in which the energy {\em decreases} is only about $6.6\cdot
10^{-7}$.
This means that the probability to find a saddle point (increase and
decrease of energy possible) for $n=3282$ is only about $0.002$. In
other words, only about 1 in 500 of the sampled states is not a local minimum,
but actually a saddle.  One can understand these numbers by
observing that flipping a link in a region where the four affected
nodes have the ``right'' degree (namely, either 5 or 7, depending on
color) will cost 4 units of energy. Since most nodes have this
property in the stationary state, a gain of four is the normal
situation. The much rarer other energy changes are possible if a link
is flipped in a region with a defect, and those are very rare. It
should be possible to quantify all this as a function of temperature,
\ie, as a function of the density of defects.

\begin{figure}[!htb]
  \centerline{\psfig{file=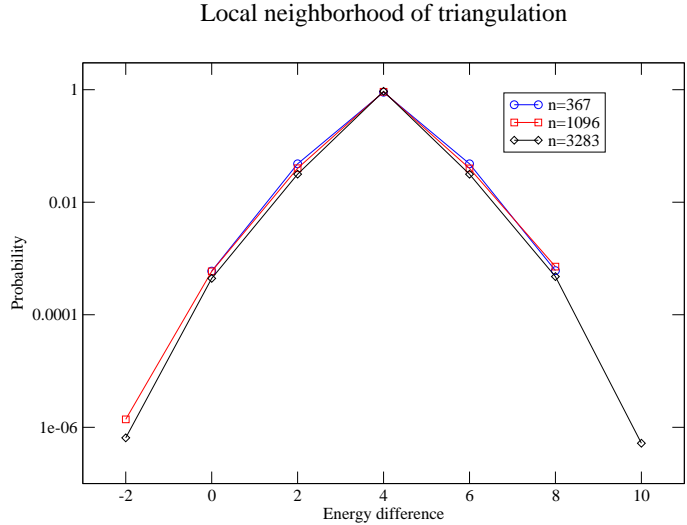,width=\x,angle=270}}
\caption{The neighborhood of typical glassy configurations. See the text for explanations.}
\label{f:neighborhoods}
\end{figure}

One can push this picture somewhat further and show that the local
minima are relatively deep. Indeed, to study the local neighborhood we
looked at all possible movements from the current state to its $3n-6$
neighbors (or slightly less when there are tetrahedra around). Note
now that the next choice of a link for the next motion will, with high
probability, affect nodes which were \emph{not} touched in the first
move. Therefore, with high probability, there is another increase of
energy by 4 units for this second step. This can go on for many more
flips. For example, if we do another $n^{1-\epsilon} $ flips, with
$\epsilon >0$, then they will all imply (with high probability) new
nodes, and each such step will increase the energy by 4. Thus, the
local minima are in troughs at least $4 n^{1-\epsilon }$ deep (with
very few directions with less increase). So the local minima are
surrounded by walls at least $4 n^{1-\epsilon }$ high, in most
directions. Note that this fact is intimately related to the topology
of the graph $\GG$.

Thus, the following picture emerges, leading to the
familiar ultrametric scenario. Any state with very few defects is
basically a local minimum. Only very well-chosen  ``exit'' directions
from such a 
state do not increase the energy. Going two steps away from the
original state, the probability of finding an exit without energy
increase is approximately the square of that finding such an exit when
doing one state. And this picture will repeat for a number of $o(n)$
steps, so that each local minimum is surrounded by walls of height
$o(n)$ and exits of probability $\OO(1/n)^k$ without increase of
energy in $k$ steps  (and probability $\OO(1/n)^{k-\ell}$ for an
energy increase of $4\ell$). Note that these observations depend only on the
short range behavior of the energy function, in our special case, the
constant 4, but not on the large scale growth of, say $(d_i-5)^2$. One
expects that these local minima will become equilibrium states, and
this is how the ultrametric property appears in this model.

\subsect{Temporal correlations}{ss:correlations}

Here we present some measurements of temporal correlations. By this one
means that one compares the triangulation at time $t$ to triangulations
at time $t+\delta t$ (it is well-known that this is the right aging
approach, see \eg, \cite{benarous2006}. The distance $D(\T,\T')$ between two
triangulations $\T$ and $\T'$ is
defined as the number of (numbered) nodes which have different degrees or
different neighbors.\footnote{More precisely, we represent the
  triangulations by fixing a certain numbering of nodes, and by
  enumerating for each node its neighbors in counter clockwise
  order. When doing comparisons, we compare these representations, as
  obtained in the simulations.}
This measure is mathematically not quite right,
since two triangulations which only differ in a renumbering of the
nodes (respecting color) would be considered equal in $\TT_n$ but
unequal here. The advantage of the current definition is that it is
very easy to implement. (A better measure would be the minimum of the
distance over all permutations of the numbering of the nodes, or the
shortest distance between $\T$ and $\T'$.) These quantities are illustrated in
Figs.~\ref{f:correlation.8_0.175} and \ref{f:correlationfromfixedtime}.

 \begin{figure}[!htb]
  \centerline{\psfig{file=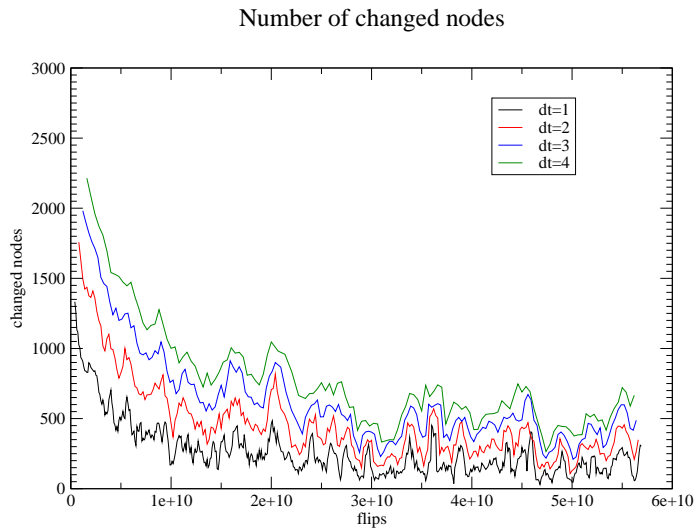,width=\x,angle=270}}
\caption{Number of nodes which \emph{differ} between time $t$ (in flips) and
  time $t+10^8\cdot dt$. The data are for $n=3283$ and $T=0.175$.}
\label{f:correlation.8_0.175}
 \end{figure}
 \begin{figure}[!htb]
  \centerline{\psfig{file=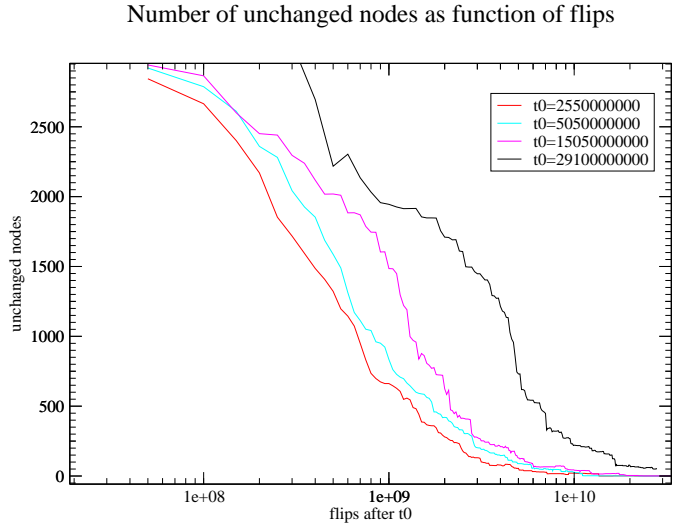,width=\x,angle=270}}
\caption{Number of nodes which are \emph{unchanged} as a function of
  $dt$. The vertical axis is $3283-D(t_0,t_0+dt)$, with $t_0$ given in the
  caption and $D$ defined in the text. The data are for $n=3283$ and
  $T=0.175$.} 
\label{f:correlationfromfixedtime}
 \end{figure}

\subsect{Spatial correlations}{ss:scorrelations}

Here, we compare two spatial correlations, one for the regular torus
triangulation, and the other for a state of the triangulation taken
towards the end of the run. These correlations are measured with a
technique known from quantum gravity (with dynamical triangulations 
of $S^4$), see \cite{bakkersmit1997}. The correlation function $C(r)$
at distance $r$ is
defined as 
\begin{equ}
C(r)= \frac{\sum_{ij: {\rm dist}(i,j)=r} (d_i-\bar d)(d_j -\bar
 d)}{\sum_{ij: {\rm dist}(i,j)=r}1}~, 
\end{equ}
where $d_i$ is the degree at $i$, $\bar d$ is the mean degree, and the
distance between two nodes $i$ and $j$ is defined as the minimal
number of hops to get from $i$ to $j$.
The power spectrum is then the amplitude of the Fourier transform of this quantity.

In \fref{f:scorrelations} we show that the
power spectrum of the regular triangulation has, as expected, a peak,
while the one for the glassy phase shows no structure at all. The
precise data are as follows: The torus triangulation is regular as
described in Remark~\ref{r:torus}, with 0 energy, and 3600 nodes.
The glassy triangulation is a typical state of a simulation done with
3283 nodes, at temperature $T=0.175$.

\begin{figure}[!htb]
 \centerline{\psfig{file=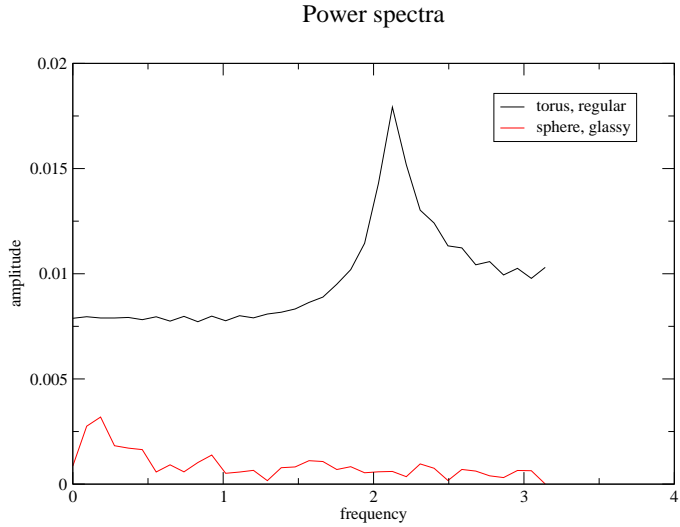,width=\x,angle=270}}
\caption{Power spectra for the torus triangulation and a glassy
  triangulation, corresponding to the equilibrium at $T=0.175$.} 
\label{f:scorrelations}
\end{figure}

\sect{A random walk interpretation}{s:nature}
Glass models can be classified in largely two different classes, and,
at present it is not clear whether these two classes are the same or are
different. The first class can be called the ``deep valley''
class. People imagine a landscape with increasingly deep valleys and
the random walk enters them, and, the deeper the valley one finds, the
harder it is to get out of the valley to find eventually an even
deeper one. The second class can be called ``narrow corridors''
class. Here, the valleys are rather flat, and, while the shortest
distance between two points might be quite short, it might be very
difficult to find a path which has a small total variation in height,
going from one configuration to another. The model in \cite{Proglass2007} and
this one seem to be of this second class.

How should one view our the model at very low temperature? Most odd (blue) nodes
will have a degree 7 while most even (red) ones have degree 5, if the
temperature is low enough. Furthermore, there will be a density of
``defects'' that is, odd ones of degree different from 7 and even ones
of degree different from 5, as seen in Table~\ref{tab:1}.
Experimentally, what happens is that at
low temperature only 6's and 8's occur for the odd ones and 4's and 6's
for the even ones. We can thus view these defects (which all cost
energy 1 each) as a gas of low density. At a given temperature, 4
defects are usually created by flipping an edge in a region with no
defects (all 4 corners will acquire a ``wrong'' degree). One new
defect can be created when one flips a link which connects to only one
old defect. The
probability of this happening (per flip) is proportional to
$\exp(-4\beta)$ (resp.~$\exp(-2\beta)$.)   Thus, 4- (or 2-) tuples are randomly created at this rate. On
the other hand, defects can wander (painfully) through the
triangulations as we have shown above in \sref{s:defects}.
When 2 or 4 of them meet they can annihilate, and the final density of defects
as a function of $\beta $ should be obtainable as the equilibrium
between creation and annihilation of these defects. Note that, since
annihilation lowers the energy, this will happen with a rate 1
whenever they meet, while creation happens with the much smaller rate
$\exp(-4\beta)$. 

\begin{remark}
Our discussion of defects differs from that of \cite{Proglass2007}. In
that paper, most particles live, at low temperature, in a hexagon.
Any red particle in a pentagon (or blue particle in a 7-gon) is then
called a \emph{glass-defect}, while all other cases are called
\emph{liquid-defects}. In contrast, in our model the natural thing is
to have red particles in pentagons and blue ones in 7-gons, and
defects are any coordination numbers different from 5 or 7. In
particular, a hexagon is a defect in our model. Given our earlier
discussion, at the temperatures we consider, all defects which appear
in the simulations would correspond to glass-like defects. The
following discussion can the be seen as a variant of
Eqs.~(2) and (3) in \cite{Proglass2007}.
\end{remark}
\subsection{A toy model}\label{ss:toy}

One can study the density of defects in a simplified model which is basically
exactly solvable\footnote{Yuval Peres and Bernard Derrida kindly
  explained to me how one discusses such models, and also suggested
  the precise law.}. The model is as follows: Take a square lattice
(sublattice of $\integer^2$) of
size $N\times N$ with periodic boundary conditions. Each site of the
lattice can be either empty or filled with one particle. Fix a
constant $\rho$ (this mimics $\exp(-4\beta)$). The Markov process
consists in choosing at random one of the sites. 
\begin{enumerate}
\item{}If it is filled, move
the particle randomly in one of the 4 directions to the next site. If
the target site is occupied, the particles annihilate each other. If not, the
particle stays at the new site.
\item{}If the site is empty, create a new particle there with
  probability $\rho$.
\end{enumerate}
The conjecture is that the equilibrium density of the particles, 
for $N\to\infty$ and small $\rho$, should behave like
\begin{equ}[asymptotic]
\const \rho^{1/2} | \log \rho |^{1/2}~.
\end{equ}
\begin{remark}
  We have checked this law for $N=100$, with a very good fit (see \fref{f:rhologrho}).
\end{remark}
\begin{figure}[!htb]
    \centerline{\psfig{file=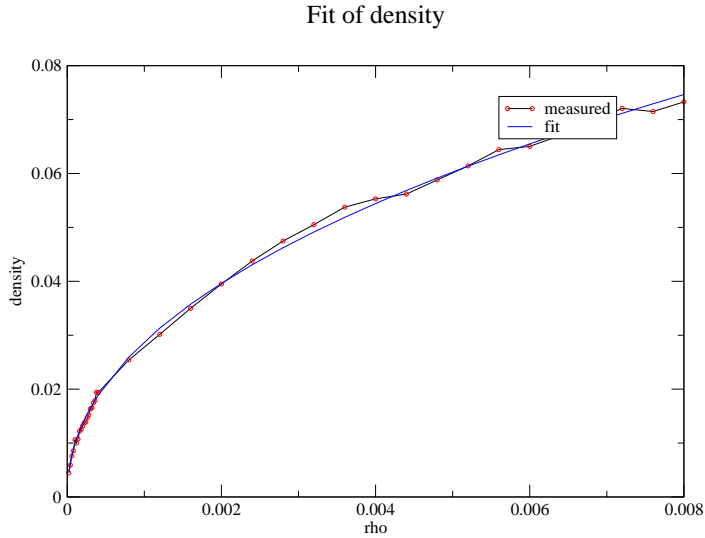,width=8.2cm,angle=270}}
\caption{The asymptotic density of walkers for the model of \sref{ss:toy} as a
 function of the temperature. The theoretical curve is
 $f(\rho)=a\rho^{1/2}|\log(b\cdot\rho)|^{1/2}$ with $a=0.252$ and $b=0.00217$.
 }
\label{f:rhologrho}
\end{figure}
The connection between this model and our model of a glass is almost obvious.
The migration of defects was discussed in \sref{s:defects}, see also
\fref{f:movedefect}. The only difference here is that moves in the
glass model are slower, since perhaps the energy will increase on the
way from a position to the next. But this only changes the time scale
of the moves of defects. The creation of defects takes usually place
either in a region where there is no defect nearby, and then the
energy increases by 4, hence the probability of this happening will be
$\exp(-4\beta)$. But perhaps other such creations will only need
energy 2, and this is not covered by the toy model. The toy model is
on a lattice $\integer^2$ while the defect model is on the triangulation
---not on the set $\TT_n$ of triangulations---since we talk here about
motion of defects, viewed as \emph{independent}, unless they
collide. Therefore, what can be reasoned on $\integer^2$ transposes to
the triangulation, since both are locally transient. Therefore, we
conjecture that for the topological glass model, the density of
defects should behave like \eref{asymptotic}, with $\rho$ of the form
$\rho =\exp(-\beta C)$, for some $C$. We have not
been able to verify this in the simulations.

Coalescing and annihilating random walks are discussed in various
places, see \eg, \cite{bergkesten2000}.

\sect{Conclusions and outlook}{conclusion}

In this paper, we have discussed a variety of properties of a
glass-like model. These properties show that the glass-like behavior
can be obtained without reference to position, but already in a
discrete phase
space (given by the graph $\GG$ of triangulations and their
connections through flips). Furthermore, the energy landscape and its
concomitant slowing down of motion to equilibrium, seem to depend
mostly only on the cost of \emph{local} energy changes, and are thus
universal. The \emph{global} structure is in fact hard-wired into the
graph $\GG$. It would be interesting to see whether the equilibrium
states can be mapped back into a physical space, for example by
mapping the triangulation onto the disk in such a way that every point
is away from every other point by at least the same minimal distance
$r$, and to compare the result to those obtained with classical potentials.

\begin{acknowledgement}
We have profited from very useful discussions with G. Ben Arous,
B. Derrida, Th. Giamarchi, N. Linial, Y. Peres, and I. Procaccia. This
work was partially 
supported by the Fonds National Suisse. 
\end{acknowledgement}
\bibliographystyle{JPE}
\markboth{\sc \refname}{\sc \refname}
\bibliography{refs}
\end{document}